\documentclass[12pt,onecolumn,oneside,a4paper,peerreview]{IEEEtran}

\usepackage{cite}
\usepackage{graphicx}
\usepackage{amsmath}
\usepackage{times}
\usepackage{latexsym}
\usepackage{graphicx}
\usepackage{bm}
\usepackage{amssymb}
\usepackage[center]{caption2}
\usepackage{stfloats}
\usepackage{cases}
\usepackage{url}
\usepackage{array}
\usepackage{setspace}
\usepackage{fancyhdr}
\usepackage{color}
\usepackage{subfigure}
\usepackage{chemarrow}
\usepackage{extarrows}

\newtheorem{theorem}{Theorem}
\newtheorem{lemma}{Lemma}
\newtheorem{corollary}{Corollary}
\newtheorem{proposition}{Proposition}

\def\proof{\noindent\hspace{2em}{\itshape Proof: }}
\def\endproof{\hspace*{\fill}~$\square$\par\endtrivlist\unskip}

\begin{document}
\title{Sum-Rate and Power Scaling of Massive MIMO Systems with Channel Aging}
\author{Chuili Kong, Caijun Zhong, Anastasios K. Papazafeiropoulos, Michail Matthaiou, and Zhaoyang Zhang
\thanks{C. Kong, C. Zhong and Z. Zhang are with the Institute of Information and Communication Engineering, Zhejiang University, China. In particular, C. Zhong is also affiliated with the National Mobile Communications Research Laboratory, Southeast University, Nanjing, China (email: caijunzhong@zju.edu.cn).}
\thanks{Anastasios K. Papazafeiropoulos is with the Department of Communications and Signal Processing Group, Imperial College London, London, U.K. (email: a.papazafeiropoulos@imperial.ac.uk).}
\thanks{M. Matthaiou is with the School of Electronics, Electrical Engineering and Computer Science, Queen's University Belfast, Belfast, U.K. He is also affiliated with the Department of Signals and Systems, Chalmers University of Technology, 412 96, Gothenburg, Sweden (email:
m.matthaiou@qub.ac.uk).}}

\maketitle

\begin{abstract}
This paper investigates the achievable sum-rate of massive multiple-input multiple-output (MIMO) systems in the presence of channel aging. For the uplink, by assuming that the base station (BS) deploys maximum ratio combining (MRC) or zero-forcing (ZF) receivers, we present tight closed-form lower bounds on the achievable sum-rate for both receivers with aged channel state information (CSI). In addition, the benefit of implementing channel prediction methods on the sum-rate is examined, and closed-form sum rate lower bounds are derived. Moreover, the impact of channel aging and channel prediction on the power scaling law is characterized. Extension to the downlink scenario and multi-cell scenario are also considered. It is found that, for a system with/without channel prediction, the transmit power of each user can be scaled down at most by $1/\sqrt{M}$ (where $M$ is the number of BS antennas), which indicates that aged CSI does not degrade the power scaling law, and channel prediction does not enhance the power scaling law; instead, these phenomena affect the achievable sum-rate by degrading or enhancing the effective signal to interference and noise ratio, respectively.
\end{abstract}

\newpage
\section{Introduction}
In order to meet the exponential growth of mobile and wireless data traffic, the fifth generation wireless systems are expected to deliver a thousand-fold higher capacity \cite{J.G.Andrews}. Among various potential enabling technologies to tackle such challenges, massive MIMO \cite{T.L.Marzetta}, where the BS deploys an unprecedented number of antennas to simultaneously serve a much smaller number of users, stands out as a promising candidate because of its remarkable capability of substantially improving both the spectral and energy efficiency \cite{E.G.Larsson,F.Rusek}. As such, massive MIMO technology has attracted enormous research attention from both academia and industry.

The gains of massive MIMO systems were initially demonstrated by assuming an ideal propagation environment. As such, understanding the performance limits of massive MIMO systems in realistic propagation environments is of paramount importance. Thus far, the impact of various practical channel imperfections on the performance of massive MIMO systems has been studied in literature by including line-of-sight effect \cite{Q.Zhang,C.Kong}, spatial correlation \cite{J.Hoydis,J.Zhang,H.Ngo0,C.Masouros}, pilot contamination \cite{J.Jose,N.Krishnan2}, pilot design for channel estimation \cite{J.Choi1,S.Noh}, channel estimation error \cite{H.Q.Ngo,C.K.Wen}, channel quantization \cite{J.Li,D.Ying,J.Choi2,J.Choi3}, transceiver hardware impairments \cite{Emil2,X.Zhang0}, and phase noise drift \cite{A.Pitarokoilis}.

In addition to the above mentioned channel/system imperfections, there is another important aspect of practical channel impairments known as {\it channel aging}; this refers to the phenomenon that channel varies between when it is learned via estimation and when it is used for precoding or detection because of the random fluctuation of the propagation channel due to the relative movement between the users and the BS, as well as, the processing delay at the BS. Despite its significance, very few works have investigated its impact on the performance of massive MIMO systems. Capitalizing on the deterministic equivalent analysis framework \cite{J.Hoydis}, the effect of inaccurate CSI due to channel aging was first studied in \cite{K.T.Truong1} by assuming matched filter at the BS. Later on, the analysis was extended to the scenario with more sophisticated receivers, such as regularized ZF precoders (downlink) \cite{A.K.Papazafeiropoulos1} and minimum-mean-square-error (MMSE) receivers (uplink) \cite{A.K.Papazafeiropoulos2}.

The analytical expressions developed in \cite{K.T.Truong1,A.K.Papazafeiropoulos1,A.K.Papazafeiropoulos2} are derived by employing the deterministic equivalent approach which relies on the key assumption of large system regime, i.e., $M \rightarrow \infty$ and $K \rightarrow \infty$, where $M$ is the number of BS antennas and $K$ is the number of users, and they only serve as accurate approximations. Hence, it is also of great interest to find tight sum rate bounds valid for arbitrary finite $M$ and $K$, which provide an alternative perspective of quantifying the sum rate. In this regard, we propose tractable and tight lower bounds on the achievable sum rate of the system. Another major limitation of the expressions in \cite{K.T.Truong1,A.K.Papazafeiropoulos1,A.K.Papazafeiropoulos2} is that they are in general too complicated to yield any useful insights into the impact of channel aging on the system performance. Motivated by this, we derive simple and informative power scaling laws which shed light into how the channel aging affects the achievable rate. In addition to the multi-cell scenario considered in \cite{K.T.Truong1,A.K.Papazafeiropoulos1,A.K.Papazafeiropoulos2}, the single cell scenario is also studied in detail in the current paper, mainly motivated by the following reasons: 1) Compared to the multi-cell scenario, the single-cell scenario provides more engineering insights as reported in many prior works \cite{L.You,Emil5,Y.Lim,A.Muller}; 2) The analytical approach developed for the asymptotic analysis of the single-cell scenario could also be applied for the multi-cell scenario; 3) With a relatively large frequency reuse factor, the single cell performance can be actually attained \cite{H.Q.Ngo}; 4) In practice, single cell massive MIMO deployment has also been considered for indoor scenarios, see for instance \cite{B.Panzner}.

Specifically, the main contributions of the paper are outlined as follows:
\begin{itemize}
  \item We obtain tight lower bounds on the achievable sum-rate of single-cell uplink massive MIMO systems employing MRC or ZF receivers with channel aging, which are valid for arbitrary number of BS antennas $M$ and number of users $K$, thereby enabling efficient evaluation of the achievable sum-rate in the presence of aged CSI.
  \item Taking into consideration channel prediction, we derive tight lower bounds on the sum-rate of single-cell uplink
massive MIMO systems employing MRC and ZF receivers.
  \item For both scenarios with/without channel prediction, we characterize the power scaling law of the system. It is shown that channel aging does not reduce the power scaling law, and using channel prediction method does not improve the power scaling law.
\item Finally, we extend the power scaling law analysis to the single-cell downlink and multi-cell uplink scenarios. It turns out that the single-cell downlink case achieves the same power scaling law, while the multi-cell uplink scenario exhibits a different power scaling law due to the pilot contamination effect.
    \end{itemize}

The remainder of this paper is organized as follows: In Section II, we describe the system model incorporating the combined effects of channel estimation error and channel aging. Section III presents the achievable uplink rate with aged CSI for MRC and ZF receivers in the single-cell uplink scenario. In Section IV, the achievable uplink rate with predicted CSI is studied. Then, Section V extends the analysis to the single-cell downlink scenario. The multi-cell scenario is considered in Section VI. Numerical results are provided in Section VII. Finally, Section VIII gives a brief summary.

{\it Notation}: We use bold upper case letters to denote matrices, bold lower case letters to denote vectors and lower case letters to denote scalars. Moreover, $(\cdot)^{\dag}$, $(\cdot)^{*}$, $(\cdot)^{T}$, and $(\cdot)^{-1}$ represent the conjugate transpose operator, the conjugate operator, the transpose operator, and the matrix inverse, respectively. Also, $|| \cdot ||$ is the Euclidian 2-norm, $| \cdot |$ is the absolute value, and ${\left[ {\bf{A}} \right]_{mn}}$ gives the $(m,n)$-th entry of $\bf{A}$. In addition, ${{\cal CN} (0,1)}$ denotes a scalar complex circular Gaussian random variable with zero mean and unit variance, while ${{\bf{I}}_k}$ is the identity matrix of size $k$. Finally, the statistical expectation operator is represented by ${\rm E}\{\cdot\}$, while the trace operator and the Kronecker product are denoted by ${\text{tr}}(\cdot)$ and $\otimes$, respectively.
\section{Single-Cell Uplink Model}
We start with the uplink of a single-cell MIMO system, which is composed of a central BS with $M$ antennas and $K$ $(K \leq M)$ noncooperative users with single antenna each. It is assumed that the propagation channel exhibits flat fading, and the channel coefficients do not change within one symbol, but vary slowly from symbol to symbol as in \cite{K.T.Truong1}. Therefore, for the \emph{n}-th symbol, the $M \times 1$ received signal at the BS is given by
\begin{align}\label{eq:y_ul}
  {\bf y}[n] = \sqrt{p_u}{\bf G}[n]{\bf x}[n] + {\bf z}[n],
\end{align}
where ${\bf G}[n]$ represents the $M \times K$ channel matrix between the BS and the $K$ users, i.e., $g_{mk}[n] = \left[{\bf G}[n]\right]_{mk}$ denotes the channel coefficient of the communication link between the \emph{m}-th antenna of the BS and the \emph{k}-th user; $p_u$ is the average transmit power of each individual user;  ${\bf x}[n]$ is a $K \times 1$ vector consisting of the transmit symbols of $K$ users with unit power; and ${\bf z}[n]$ is the zero-mean additive white Gaussian noise with unit variance.

The channel coefficient $g_{mk}[n]$ can be written as
 \begin{align}
  g_{mk}[n] = h_{mk}[n] \sqrt{\beta_k},
\end{align}
where $h_{mk}[n]$ is the small-scale fading coefficient for the link from the \emph{k}-th user to the \emph{m}-th antenna of the BS, which is assumed to be independent and identically distributed (i.i.d.) ${{\cal CN} (0,1)}$, and $\beta_k$ models the large-scale effect including shadowing and pathloss, which is assumed to remain constant for all $n$. Hence, ${\bf G}[n]$ can be expressed in a matrix form as
\begin{align}
  {\bf G}[n] = {\bf H}[n]{\bf D}^{\frac{1}{2}},
\end{align}
where ${\bf H}[n]$ is an $M \times K$ matrix with $[{\bf H}[n]]_{mk} = h_{mk}[n]$, and $\bf D$ is a $K \times K$ diagonal matrix with $[{\bf D}]_{kk} = \beta_k$.
\subsection{Channel Estimation}
The BS estimates the channels using uplink pilots. Let $\tau$ be the length of the training period, then, the pilot sequences used by the $K$ users can be represented by a $K \times \tau$ matrix $\sqrt{p_p}{\bf \Phi}$ ($\tau \geq K$) satisfying ${\bf \Phi}{\bf \Phi}^\dag = {\bf I}_K$, where $p_p \triangleq \tau p_u$. Therefore, the $M \times \tau$ received pilot matrix at the BS is given by \cite{K.T.Truong1,A.K.Papazafeiropoulos1,A.K.Papazafeiropoulos2},
\begin{align}
  {\bf J}[n] = \sqrt{p_p}{\bf G}[n]{\bf \Phi} + {\bf {\tilde Z}}[n],
\end{align}
where ${\bf N}[n]$ is an $M \times \tau$ noise matrix whose elements are i.i.d. ${\cal {CN}}(0,1)$. To estimate ${\bf G}[n]$, the BS first correlates ${\bf J}[n]$ with ${\bf \Phi}^\dag$ to obtain
\begin{align}
  {\tilde{\bf Y}}[n] = \frac{1}{\sqrt{p_p}}{\bf J}[n]{\bf \Phi}^\dag,
\end{align}
which gives the following observation of the channel vector from user $k$ to the BS
\begin{align}
  {\tilde {\bf y}}_{k}[n] = {\bf g}_k[n] + \frac{1}{\sqrt{p_p}}{\bf b}_k[n],
\end{align}
where ${\bf g}_k[n]$ and ${\bf b}_k[n]$ are the \emph{k}-th columns of the matrices ${\bf G}[n]$ and ${\bf B}[n]\triangleq {\bf {\tilde Z}}[n]{\bf \Phi}^\dag$, respectively. Since ${\bf \Phi}{\bf \Phi}^\dag = {\bf I}_K$, ${\bf B}[n]$ has i.i.d ${\cal CN}(0,1)$ elements.

As in \cite{H.Q.Ngo}, the MMSE estimate of ${\bf G}[n]$, given $\tilde{\bf Y}[n]$, is
\begin{align}\label{eq:G_hat}
  {\hat{\bf G}}[n] = {\tilde{\bf Y}}_p[n]{\tilde {\bf D}} = \left({\bf G}{[n]} + \frac{1}{\sqrt{p_p}}{\bf W}[n]\right){\tilde {\bf D}},
\end{align}
where ${\tilde {\bf D}} \triangleq (\frac{1}{p_p}{\bf D}^{-1} + {\bf I}_K)^{-1}$. As such, ${\bf g}_k[n]$ can be decomposed into
\begin{align}\label{eq:g_esti_error}
  {\bf g}_k[n] = {\hat {\bf g}}_k[n] + {\tilde {\bf g}}_k[n],
\end{align}
where ${\hat {\bf g}}_k[n]$ is the \emph{k}-th column of ${\hat {\bf G}}[n]$, and ${\tilde {\bf g}}_k[n]$ is the estimation error vector for the \emph{k}-th user. After some simple algebraic manipulations based on \eqref{eq:G_hat}, it can be shown that each element of ${\hat {\bf g}}_k[n]$ is a Gaussian random variable with zero mean and variance $\frac{p_p \beta_k^2}{1+p_p \beta_k}$. Furthermore, ${\hat {\bf g}}_k[n]$ and ${\tilde {\bf g}}_k[n]$ are independent due to the orthogonality property of linear MMSE estimators. At this point, it is worth mentioning that there are different types of channel error models, i.e., unbounded error (usually modeled as Gaussian distributed, such as the one considered here) and bounded error (such as ball or ellipsoid error, see references \cite{N.Vucic,M.F.Hanif,E.Bjornson}); Also, the ellipsoid error model considered in \cite{N.Vucic,M.F.Hanif,E.Bjornson} can mathematically correspond to the Gaussion error vector given by the second term in \eqref{eq:g_esti_error}.
\subsection{Channel Aging}
In practice, due to the random fluctuations of the propagation caused by the movement of users and the processing delays at the BS, the channel varies between when it is learned via estimation and when it is applied for precoding or detection. Such phenomenon is referred to as {\it channel aging} in the literature. To investigate the impact of channel aging, we adopt the model proposed in \cite{K.T.Truong1}. As such, the $M \times 1$ channel vector for the \emph{k}-th user at time $n+1$ can be expressed through an autoregressive model of order 1 as
\begin{align}\label{eq:g_chan_aging}
  {\bf g}_k[n+1] = \alpha {\bf g}_k[n] + {\bf e}_k[n+1],
\end{align}
where ${\bf e}_k[n+1]$ is a temporally uncorrelated complex white Gaussian noise process with its elements having variance of $\left( 1- \alpha^2 \right) \beta_k$, and $\alpha$ is a temporal correlation parameter. Considering the Jakes fading model, we have $\alpha = J_0(2\pi f_D T_s)$, where $J_0(\cdot)$ is the zero-order first kind Bessel function, $T_s$ is the channel sampling duration, $f_D$ is the maximum Doppler frequency shift determined by the users' velocity $v$ and carrier frequency $f_c$, as $f_D = \frac{v f_c}{c}$ ($c$ denotes the speed of light). From the properties of the Bessel function, we can easily get $0 \leq |\alpha| \leq 1$. Especially, the smaller $|\alpha|$, the more serious the channel aging effect becomes.

To this end, a model accounting for the combined effects of the channel estimation errors and channel aging effect can be expressed as \cite{K.T.Truong1}
\begin{align}\label{eq:g_esti_ca}
  {\bf g}_k[n+1] &= \alpha {\hat {\bf g}}_k[n] + \underbrace{\alpha {\tilde {\bf g}}_k[n] + {\bf e}_k[n+1]}_{{\tilde {\bf e}}_k[n+1]},
\end{align}
where ${{\tilde {\bf e}}_k[n+1]}$ is independent with ${\hat {\bf g}}_k[n]$ due to the independence between ${\tilde{\bf g}}_k[n]$, ${\bf e}_k[n+1]$, and ${\hat{\bf g}}_k[n]$. As a result, each element of ${{\tilde {\bf e}}_k[n+1]}$ is a complex Gaussian random variable with zero mean and variance $\beta_k - \alpha^2\frac{p_p\beta_k^2}{1+p_p\beta_k}$.
\section{Achievable Uplink Sum-Rate with Channel Aging}\label{sec:section:3}
In this section, we present a detailed analysis of the impact of channel aging on the achievable sum-rate of the system with linear receivers. In particular, two popular linear receivers, namely, MRC and ZF receivers are considered. For both receivers, we derive closed-form lower bounds of the achievable sum-rate with aged CSI. Moreover, the impact of aged CSI on the asymptotic power scaling law is characterized.

As in \cite{K.T.Truong1}, we assume that the large-scale effect $\bf D$ and the temporal correlation parameter $\alpha$ are known at the BS.\footnote{In practice, the large-scale effect varies much slower. Hence, it can be easily estimated by the BS. In addition, given that the velocity of the users can be obtained by the BS,  the temporal correlation parameter $\alpha$ can be accurately estimated by the BS.} Hence, the BS has the following CSI
\begin{align}\label{eq:g_bar_aged}
  {\bar {\bf g}}_k[n+1] = \alpha {\hat {\bf g}}_k[n].
\end{align}

Let ${\hat{\bf A}}[n+1]$ be an $M \times K$ linear detector matrix which depends on the channel ${\bar {\bf G}}[n+1]$, where ${\bar {\bf G}}[n+1] \triangleq [{\bar {\bf g}}_1[n+1], {\bar {\bf g}}_2[n+1], \cdots, {\bar {\bf g}}_K[n+1]]$. By considering linear receivers, the received signal is separated into streams by multiplying ${\hat{\bf A}}^\dag[n+1]$ with ${\bf y}[n+1]$ from \eqref{eq:y_ul} as follows
\begin{align}\label{eq:r_ul}
  {{\bf r}}[n+1] &= {\hat{\bf A}}^\dag[n+1] {\bf y}[n+1]\nonumber\\
  &= \sqrt{p_u}{\hat{\bf A}}^\dag[n+1]{\bf G}[n+1]{\bf x}[n+1] + {\hat{\bf A}}^\dag[n+1]{\bf z}[n+1].
\end{align}

We consider two conventional linear receivers, i.e., MRC and ZF, which are expressed as
\begin{align}
  {\hat{\bf A}}[n+1] =
  \begin{cases}
    {\bar{\bf G}}[n+1], & \text{for MRC},\\
    {\bar{\bf G}}[n+1]\left({\bar{\bf G}}^\dag[n+1] {\bar{\bf G}}[n+1]\right)^{-1}, & \text{for ZF}.
  \end{cases}
\end{align}

Moreover, let ${{r}}_k[n+1]$ and $ x_k[n+1]$ be the \emph{k}-th elements of the $K \times 1$ vectors ${{\bf r}}[n+1]$ and ${\bf x}[n+1]$, respectively. Then, from \eqref{eq:r_ul}, the \emph{k}-th element of ${\bf r}[n+1]$ is given by
\begin{align}\label{eq:r}
  {r}_k[n+1]
  &= \sqrt{p_u}{\hat{\bf a}}_k^\dag[n+1]{\bar{\bf g}}_k[n+1]{x}_k[n+1]
  + \sqrt{p_u} \sum\limits_{i = 1, i \neq k}^K {{\hat{\bf a}}_k^\dag[n+1]{\bar{\bf g}}_i[n+1]{x}_i[n+1]}\notag \\
  &+ \sqrt{p_u} \sum\limits_{i = 1}^K {{\hat{\bf a}}_k^\dag[n+1]({{\bf g}}_i[n+1] - {\bar{\bf g}}_i[n+1]){x}_i[n+1]}
  + {\hat{\bf a}}_k^\dag[n+1]{\bf z}[n+1],
\end{align}
where ${\hat{\bf a}}_k[n+1]$ is the \emph{k}-th column of ${\hat{\bf A}}[n+1]$. The BS treats ${\bar{\bf g}}_k[n+1]$ as the true channel for the \emph{k}-th user, and the part including the last three terms of \eqref{eq:r} is considered as interference plus noise. As in \cite{Q.Zhang,H.Q.Ngo,J.Hoydis}, the combined error ${\bf g}_i[n+1]-\bar{\bf g}_i[n+1]$ is treated as uncorrelated Gaussian noise,
which is a worst-case scenario, therefore leading to the following simple lower bound for the achievable rate of the $k$-th user: \begin{multline}\label{eq:R_k}
 R_k = \\{\rm E}\left\{ { \log_2 \left(1+ \frac{p_u|{\hat{\bf a}}_k^\dag[n+1]{\bar{\bf g}}_k[n+1]|^2}{p_u \sum\limits_{i = 1, i \neq k}^K {|{\hat{\bf a}}_k^\dag[n+1]{\bar{\bf g}}_i[n+1]|^2}
 + ||{\hat{\bf a}}_k[n+1]||^2 \left(p_u\sum\limits_{i = 1}^K {\left(\beta_i - \alpha^2\frac{p_p\beta_i^2}{1+p_p\beta_i}\right)} + 1 \right)} \right)} \right\},
\end{multline}
where the expectation is taken over small-scale fading. Note that the advantage of the expression in (\ref{eq:R_k}) is that it is amenable to algebraic manipulations.

In the sequel, $R_k$ is referred to as the achievable rate of the $k$-th user. Then, the achievable sum-rate of the massive MIMO system is given by
\begin{align}
  R = \frac{T-\tau}{T}\sum_{k=1}^K R_k.
\end{align}

\subsection{MRC Receivers}
By starting with the MRC receivers, we obtain the following key result:
\begin{theorem}\label{theo:MRC_law}
   For MRC receivers, with \underline{a}ged CSI, if each user scales down its transmit power proportionally to $1/M^\gamma$, i.e., $p_u = E_u/M^\gamma$, for fixed $E_u$ and $\gamma >0$, we have
  \begin{align}\label{eq:theo:MRC_law}
  R_k^{\text{a,mrc}} - \log_2 \left(1+ \frac{\alpha^2 \tau E_u^2 \beta_k^2}{M^{2\gamma-1}} \right) \overset{M \rightarrow \infty}{\longrightarrow} 0,
  \end{align}
 where the superscript a is used to denote \underline{a}ged CSI.
\end{theorem}
\begin{proof}
Substituting ${\hat{\bf a}}_k[n+1] = {\bar{\bf g}}_k[n+1] = \alpha {\hat{\bf g}}_k[n]$ and $p_u = E_u/M^\gamma$ into \eqref{eq:R_k}, and after some simple algebraic manipulations, we obtain
\begin{align}\label{eq:R_k_2}
  R_k^\text{a,mrc} = {\rm E}\left\{ { \log_2 \left(1+ \frac{\frac{E_u}{M^\gamma}\frac{1}{M^2}\alpha^2||{\hat{\bf g}}_k[n]||^4}{\frac{E_u}{M^\gamma}\frac{1}{M^2} \alpha^2 \sum\limits_{i = 1, i \neq k}^K
{|{\hat{\bf g}}^\dag_k[n]{\hat{\bf g}}_i[n]|^2}
 + \left( \frac{E_u}{M^\gamma} \sum\limits_{i = 1}^K {\left(\beta_i - \alpha^2\frac{\tau \frac{E_u}{M^\gamma} \beta_i^2}{1+\tau \frac{E_u}{M^\gamma} \beta_i}\right)}
 + 1 \right)\frac{1}{M^2}||{\hat{\bf g}}_k[n]||^2} \right)} \right\}.
\end{align}

To this end, looking into the asymptotic large antenna regime, i.e., $M \rightarrow \infty$, and invoking the law of large numbers, we get
\begin{align}
 \frac{1}{M}|{\hat{\bf g}}^\dag_k[n]{\hat{\bf g}}_i[n]|^2 -
  \begin{cases}
  \frac{\tau \frac{E_u}{M^\gamma} \beta_k^2}{1+\tau \frac{E_u}{M^\gamma} \beta_k}, & i =k\\
   0, & i\neq k
  \end{cases}
  \overset{M \rightarrow \infty}{\longrightarrow} 0.
\end{align}

Please note, in the above derivation, we have used the fact that $\hat{\bf g}_k[n]$ and  $\hat{\bf g}_i[n]$ $(i \neq k)$ are independent, which can be easily proven according to (\ref{eq:G_hat}). We also have
\begin{align}
  \frac{E_u}{M^\gamma} \sum\limits_{i = 1}^K {\left(\beta_i - \alpha^2\frac{\tau \frac{E_u}{M^\gamma} \beta_i^2}{1+\tau \frac{E_u}{M^\gamma} \beta_i}\right)}
 + 1 \rightarrow 1.
\end{align}

Then, \eqref{eq:R_k_2} simplifies to
\begin{align}
R_k^\text{a,mrc} - \log_2 \left(1+ \frac{\alpha^2 \tau E_u^2 \beta_k^2}{M^{2\gamma-1}} \right) \overset{M \rightarrow \infty}{\longrightarrow} 0,
\end{align}
which completes the proof.
\end{proof}

Theorem \ref{theo:MRC_law} suggests that the asymptotic achievable rate $R_k^\text{a,mrc}$ depends on the choice of $\gamma$. When $\gamma>\frac{1}{2}$, $R_k^\text{a,mrc}$ converges to zero, which indicates that the transmit power of each user has been reduced too much. On the other hand, when $\gamma<\frac{1}{2}$, $R_k^\text{a,mrc}$ grows without bound, which indicates that the transmit power of each user can be scaled down further to maintain the same performance. When $\gamma = \frac{1}{2}$, $R_k^\text{a,mrc}$ converges to a non-zero limit. As such, by setting $\gamma = \frac{1}{2}$, we have the following corollary.

\begin{corollary}\label{coro:MRC_aged_CSI}
 For MRC receivers, with \underline{a}ged CSI, each user can scale down its transmit power at most by $p_u = E_u/\sqrt{M}$ for a fixed $E_u$, and the achievable uplink rate of the \emph{k}-th user becomes
  \begin{align}
    R_k^{\text{a,mrc}} \rightarrow \log_2 \left(1+ {\alpha^2 \tau E_u^2 \beta_k^2} \right), M \rightarrow \infty.
  \end{align}
\end{corollary}

Corollary \ref{coro:MRC_aged_CSI} suggests a very encouraging result, that taking into account the channel aging effect, the same power scaling law can be achieved as in those scenarios where only channel estimation error is considered \cite{H.Q.Ngo,Q.Zhang}. In other words, channel aging does not affect the power scaling law, it only leads to a reduction of the effective SINR. For the special case $\alpha=1$, i.e., no channel aging effect, the achievable rate for the \emph{k}-th user becomes $\log_2 \left(1+ { \tau E_u^2 \beta_k^2} \right)$, which agrees with the result presented in \cite[Proposition 5]{H.Q.Ngo}.

We now turn our attention to the finite $M$ regime, and present the following tight lower bound on the achievable rate of the \emph{k}-th user.

\begin{theorem}\label{theo:MRC_bound}
  For MRC receivers, with \underline{a}ged CSI, the achievable uplink rate of the \emph{k}-th user is lower bounded by ${R}_k^\text{a,mrc} \geq {\tilde R}_k^\text{a,mrc}$ with
  \begin{align}\label{eq:R_bound_mrc}
    {\tilde R}_k^\text{a,mrc} \triangleq \log_2 \left( 1+ \frac{\alpha^2 \tau p_u^2 (M-1) \beta_k^2}{p_u(1+\tau p_u \beta_k)\sum\limits_{i = 1, i \neq k}^K{\beta_i} + (\tau +1)p_u\beta_k + 1 + b_\text{mrc} } \right),
  \end{align}
  where
  \begin{align}
    b_\text{mrc} \triangleq  (1-\alpha^2)\tau p_u^2\beta_k^2.
  \end{align}
\end{theorem}
\proof See Appendix \ref{app:theo:MRC_bound}.

It is not difficult to show that ${\tilde R}_k^\text{a,mrc}$ is an increasing function with respect to $\alpha$. Now, since $\alpha$ is related to the severity level of the channel aging effect, Theorem \ref{theo:MRC_bound} actually demonstrates the intuitive result that the more severe channel aging becomes, the lower the achievable rate. If we substitute $p_u = E_u/M^\gamma$ into \eqref{eq:R_bound_mrc} as $M \rightarrow \infty$, and after some simple mathematical manipulations, we get
\begin{align}
 {\tilde R}_k^{\text{a,mrc}} - \log_2 \left(1+ \frac{\alpha^2 \tau E_u^2 \beta_k^2}{M^{2\gamma-1}} \right) \overset{M \rightarrow \infty}{\longrightarrow} 0,
\end{align}
which exactly coincides with the limit obtained from Theorem \ref{theo:MRC_law}, suggesting the asymptotic tightness of the proposed lower bound.
\subsection{ZF Receivers}
We now turn out attention to the ZF receivers for which $M\geq K$, and obtain the following key result:
\begin{theorem}\label{theo:ZF_law}
  For ZF receivers, with \underline{a}ged CSI, if each user scales down its transmit power proportionally to $1/M^\gamma$, i.e., $p_u = E_u/M^\gamma$, where $\gamma >0$ and $E_u$ is fixed, we have
  \begin{align}
 R_k^{\text{a,zf}} - \log_2 \left(1+ \frac{\alpha^2 \tau E_u^2 \beta_k^2}{M^{2\gamma-1}} \right) \overset{M \rightarrow \infty}{\longrightarrow} 0.
  \end{align}
\end{theorem}

\begin{proof}With ZF receivers, ${\hat{\bf A}}^\dag[n+1]{\bar{\bf G}}[n+1] = {\bf I}_K$, namely, ${\hat{\bf a}}_k^\dag[n+1]{\bar{\bf g}}_i[n+1] = \delta_{ki}$ where $\delta_{ki} =1$ when $k=i$ and $0$ otherwise. Based on this, substituting $p_u = E_u/M^\gamma$ into \eqref{eq:R_k}, we get
\begin{align}\label{eq:R_k_zf}
  R_k^\text{zf} = {\rm E}\left\{ { \log_2 \left(1+ \frac{M\frac{E_u}{M^\gamma}}{\left( \frac{E_u}{M^\gamma}\sum\limits_{i = 1}^K {\left(\beta_i - \alpha^2\frac{\tau \frac{E_u}{M^\gamma} \beta_i^2}{1+\tau \frac{E_u}{M^\gamma} \beta_i}\right)}
 + 1 \right)\left[\left(\frac{{\bar{\bf G}}^\dag[n+1]{\bar{\bf G}[n+1]}}{M}\right)^{-1}\right]_{k,k}} \right)} \right\}.
\end{align}

To this end, use of the law of large numbers yields
\begin{align}\label{eq:G_bar}
\left[\left(\frac{{\bar{\bf G}}^\dag[n+1]{\bar{\bf G}}[n+1]}{M}\right)^{-1}\right]_{k,k} -
\frac{1}{\alpha^2} \frac{1+\tau \frac{E_u}{M^\gamma} \beta_k} {\tau \frac{E_u}{M^\gamma} \beta_k^2} \overset{M \rightarrow \infty}{\longrightarrow} 0.
\end{align}

Now, by plugging \eqref{eq:G_bar} into \eqref{eq:R_k_zf}, the desired result can be obtained after some simple algebraic manipulations.
\end{proof}

Theorem \ref{theo:ZF_law} indicates that ZF receivers attain the same power scaling law as MRC receivers, i.e., $1/\sqrt{M}$, and achieve the same non-zero limit, which is consistent with prior results reported in \cite{H.Q.Ngo,Q.Zhang}. Hence, it can be concluded that, for ZF receivers with aged CSI, the transmit power of each user can be cut down at most by $1/\sqrt{M}$ with no rate degradation, and the achievable uplink rate is the same as that for MRC receivers.

In the finite $M$ regime, we obtain the following lower bound on the achievable rate.
\begin{theorem}\label{theo:ZF_bound}
   For ZF receivers, with \underline{a}ged CSI, the achievable uplink rate of the \emph{k}-th user is lower bounded by ${R}_k^\text{a,zf} \geq {\tilde R}_k^\text{a,zf}$ with
  \begin{align}\label{eq:R_bound_zf}
    {\tilde R}_k^\text{a,zf} \triangleq \log_2 \left( 1+ \frac{\alpha^2 \tau p_u^2 (M-K)\beta_k^2}{(1 + \tau p_u \beta_k )\sum\limits_{i = 1}^K{\frac{p_u\beta_i}{\tau p_u\beta_i +1}} + \tau p_u\beta_k + 1 + b_\text{zf} } \right),
  \end{align}
  where
  \begin{align}
    b_\text{zf} \triangleq (1-\alpha^2)(1+\tau p_u\beta_k)\sum\limits_{i = 1}^K{\frac{\tau p_u^2\beta_i^2}{1+\tau p_u\beta_i}}.
  \end{align}
\end{theorem}
\proof See Appendix \ref{app:theo:ZF_bound}.

By substituting $p_u = E_u/M^\gamma$ into \eqref{eq:R_bound_zf} as $M \rightarrow \infty$, and after some simple algebraic manipulations, we have
\begin{align}
{\tilde R}_k^{\text{zf}} - \log_2 \left(1+ \frac{\alpha^2 \tau E_u^2 \beta_k^2}{M^{2\gamma-1}} \right) \overset{M \rightarrow \infty}{\longrightarrow} 0,
\end{align}
which indicates the asymptotic tightness of the lower bound \eqref{eq:R_bound_zf}.
 \section{Achievable Uplink Rate with Channel Prediction}
As shown in the previous section, the channel aging effect results in a loss at  the achievable rate. To alleviate this undesired implication, channel prediction methods were proposed in \cite{K.T.Truong1,A.K.Papazafeiropoulos1,A.K.Papazafeiropoulos2}. In this section, we investigate the impact of channel prediction on the system performance. More specifically, closed-form lower bounds on the achievable rate are derived for both MRC and ZF receivers. In addition, the power scaling law is also characterized, based on which, the impact of the predictor order on the scaling law is evaluated.

We adopt a very popular linear predictor, i.e., the Wiener predictor proposed in \cite{K.T.Truong1}. Therefore, for the \emph{k}-th user, the channel ${\bf g}_k[n+1]$ is predicted according to ${\bar{\bf y}}_{k}[n]$, where ${\bar {\bf y}}_{k}[n] = \left[{\tilde {\bf y}}^T_{k}[n], {\tilde {\bf y}}^T_{k}[n-1], \ldots, {\tilde {\bf y}}^T_{k}[n-p]\right]^T$ with $p$ being the predictor order. To this end, we need to obtain the optimal \emph{p}-th order linear Wiener predictor, which is given in the following lemma:
 \begin{lemma}\label{lemma:opt:predictor}
   The optimal \emph{p}-th linear Wiener predictor is given by
   \begin{align}
     {\bf q}_k = \alpha \beta_k \left[ \delta \left( p,\alpha\right) \otimes {\bf I}_M \right]{\bf T}_k^{-1}(p,\alpha),
   \end{align}
   where $\delta(p,\alpha) = [1 \ \alpha \ \cdots\  \alpha^p]$, and
   \begin{align}\label{eq:T_k}
     {\bf T}_k(p,\alpha) = \beta_k \left[ \Delta \left( p,\alpha\right) \otimes {\bf I}_M \right] + \frac{1}{p_p}{\bf I}_{M(p+1)}
   \end{align}
   with
   \begin{align}
     \Delta(p,\alpha) \triangleq \left( {\begin{array}{*{20}{c}}
1&\alpha & \cdots &{{\alpha ^p}}\\
\alpha &1& \cdots &{{\alpha ^{p - 1}}}\\
 \vdots & \vdots & \ddots & \vdots \\
{{\alpha ^p}}&{{\alpha ^{p - 1}}}& \cdots &1
\end{array}} \right).
   \end{align}
 \end{lemma}
 \proof Following similar steps as in the proof of Theorem 1 in \cite{K.T.Truong1}, we can obtain the desired result. 
  \endproof

Having characterized the optimal predictor, the predicted channel can then be obtained as
 \begin{align}\label{eq:g_breve}
 {\bar{{\bf g}}}_k[n+1] = {\breve{{\bf g}}}_k[n+1] = {\bf q}_k{\bar{\bf y}}_{p,k}[n].
 \end{align}

Furthermore, the resulting mean square error between the predicted channel ${\breve{\bf g}}_k[n+1]$ and the true channel ${\bf g}_k[n+1]$ can be calculated as
\begin{align}
  \epsilon_p &= {\rm E}\{||{\bf g}_k[n+1] - {\breve{\bf g}}_k[n+1]||^2\}\\
  & \overset{(a)}{=} {\text{tr}}({\rm E}\{ ({\bf g}_k[n+1] - {\bf q}_k{\bar{\bf y}}_{p,k}[n] ){\bf g}^\dag_k[n+1]\})\\
  &= {\text{tr}}(\beta_k{\bf I}_M - \alpha^2{\bf\Theta}_k(p,\alpha)),
\end{align}
where in $(a)$ we applied the orthogonality of ${\bf g}_k[n+1] - {\breve{\bf g}}_k[n+1]$ and ${\breve{\bf g}}_k[n+1]$, and
\begin{align}\label{eq:Theta_k}
  {\bf\Theta}_k(p,\alpha) \triangleq \beta_k^2 \left[ \delta \left( p,\alpha\right) \otimes {\bf I}_M \right]{\bf T}_k^{-1}(p,\alpha) \left[ \delta \left( p,\alpha\right) \otimes {\bf I}_M \right].
\end{align}

Hence, the covariance matrix of ${\breve{\bf g}}_k[n+1]$ is given by $\alpha^2{\bf\Theta}_k(p,\alpha)$. Finally, the true channel can be decomposed as
\begin{align}
  {\bf g}_k[n+1] = {\breve{\bf g}}_k[n+1] + {\check{\bf e}}_k[n+1],
\end{align}
where ${\check{\bf e}}_k[n+1]$ is the channel prediction error vector with covariance matrix  $\beta_k{\bf I}_M - \alpha^2{\bf\Theta}_k(p,\alpha)$, which is independent of ${\breve{\bf g}}_k[n+1]$.


Before proceeding, we find it crucial to first characterize the structure of the matrix ${\bf\Theta}_k(p,\alpha)$ by the following important observation:
\begin{lemma}\label{lemma:1}
  ${\bf\Theta}_k(p,\alpha)$ is a scaled identity matrix of size $M \times M$.
\end{lemma}
\proof
Notice that ${\bf T}_k(p,\alpha)$ has the following structure
\begin{align}
  {\bf T}_k(p,\alpha) = {\bf A} \otimes {\bf I}_M,
\end{align}
where $\bf A$ is an invertible matrix, whose entries are denoted by $a_{ij}$ with $1 \leq i, j \leq n$.

Using the matrix inversion property of Kronecker product \cite[Eq. (1.10.8)]{X.Zhang}, we get
\begin{align}\label{eq:kron_property}
  ({\bf A} \otimes {\bf I}_M)^{-1} = {\bf A}^{-1} \otimes {\bf I}_M.
\end{align}

Now, let us define ${\bf B} = {\bf A}^{-1}$, where the $i,j$th element of $\bf B$ is expressed as $b_{ij}$ with $1 \leq i, j \leq n$. Hence, by combining \eqref{eq:kron_property} and \eqref{eq:Theta_k}, we obtain
\begin{align}
  {\bf\Theta}_k(p,\alpha) = \beta_k^2 \sum\limits_{i = 1}^{p+1} \sum\limits_{j = 1}^{p+1}  \alpha^{2(i-1)}b_{ij} {\bf I}_M,
\end{align}
which concludes the proof.
\endproof

Equipped with Lemma \ref{lemma:1}, it can be straightforwardly shown that the variances of the elements of ${\breve{\bf g}}_k[n+1]$ and ${\check{\bf e}}_k[n+1]$ are $\frac{1}{M}{\text {tr}}\left( \alpha^2 {\bf \Theta}_k(p,\alpha)\right)$ and $\frac{1}{M}{\text {tr}}\left( \beta_k{\bf I}_M - \alpha^2 {\bf \Theta}_k(p,\alpha)\right)$, respectively.

With predicted CSI, if we substitute ${\bar{\bf g}}_k[n+1] = {\breve{\bf g}}_k[n+1]$ into \eqref{eq:r}, we can obtain the following achievable uplink rate of the \emph{k}-th user
\begin{align}\label{eq:R_p_k}
 R^{\text p}_k = {\rm E}\left\{ { \log_2 \left(1+ {\text{SINR}}_k^{\text p} \right)} \right\},
\end{align}
where the superscript p is used to denote the \underline{p}redicted CSI, and ${\text{SINR}}_k^{\text p}$ is the signal-to-interference-noise (SINR), which is given by
\begin{align}\label{eq:SINR_p_k}
{\text{SINR}}_k^{\text p}  = \frac{p_u|{\hat{\bf a}}_k^\dag[n+1]{\breve{\bf g}}_k[n+1]|^2}{p_u \sum\limits_{i = 1, i \neq k}^K {|{\hat{\bf a}}_k^\dag[n+1]{\breve{\bf g}}_i[n+1]|^2}
 + ||{\hat{\bf a}}_k[n+1]||^2 \left( p_u \sum\limits_{i = 1}^K {\frac{1}{M}{\text {tr}}\left( \beta_i{\bf I}_M - \alpha^2 {\bf \Theta}_i(p,\alpha)\right)} + 1 \right)},
\end{align}
and the expectation in \eqref{eq:R_p_k} is taken over small-scale fading.

\subsection{MRC Receivers}
\begin{theorem}\label{theo:MRC_law_pre}
   For MRC receivers, with \underline{p}redicted CSI, if each user scales down its transmit power proportionally to $1/M^\gamma$, i.e., $p_u = E_u/M^\gamma$, where $\gamma >0$ and $E_u$ is fixed, we have
  \begin{align}\label{eq:MRC:law:pre}
R_k^{\text{p,mrc}} - \log_2 \left(1+ \frac{\alpha^2 \sum\limits_{j = 0}^p{\alpha^{2j}} \tau E_u^2 \beta_k^2}{M^{2\gamma-1}} \right)\overset{M \rightarrow \infty}{\longrightarrow} 0.
  \end{align}
\end{theorem}
\proof
By substituting ${\hat{\bf a}}_k[n+1] = {\breve{\bf g}}_k[n+1]$ and $p_u = E_u/M^\gamma$ into \eqref{eq:SINR_p_k}, and after some simple manipulations, the SINR ${\text{SINR}}_k^{\text p}$ can be expressed as
\begin{align}\label{eq:R_p_k_mrc}
 {\text{SINR}}_k^{\text p} =  \frac{\frac{E_u}{M^\gamma}\frac{1}{M^2}||{\breve{\bf g}}_k[n+1]||^4}{\frac{E_u}{M^\gamma} \sum\limits_{i = 1, i \neq k}^K {\frac{1}{M^2}|{\breve{\bf g}}^\dag_k[n+1]{\breve{\bf g}}_i[n+1]|^2}
 + \frac{1}{M^2}||{\breve{\bf g}}_k[n+1]||^2 \left( \frac{E_u}{M^\gamma} \sum\limits_{i = 1}^K {\frac{1}{M}{\text {tr}}\left( \beta_i{\bf I}_M - \alpha^2 {\bf \Theta}_i(p,\alpha)\right)}
 + 1 \right)}.
\end{align}

Since ${\breve{\bf g}}^\dag_k[n+1]$ and ${\breve{\bf g}}_i[n+1]$ $(i \neq k)$ are independent from \eqref{eq:g_breve}, we invoke the law of large numbers, when $M \rightarrow \infty$, and clearly obtain
\begin{align}\label{eq:mrc_g_limit}
  \frac{1}{M} {\breve{\bf g}}^\dag_k[n+1]{\breve{\bf g}}_i[n+1] -
  \begin{cases}
   \frac{1}{M}{\text {tr}}\left( \alpha^2 {\bf \Theta}_k(p,\alpha)\right), & i =k\\
   0, & i\neq k
  \end{cases}
  \overset{M \rightarrow \infty}{\longrightarrow} 0.
\end{align}

Hence, the remaining task is to compute ${\text {tr}}\left( \alpha^2 {\bf \Theta}_k(p,\alpha)\right)$. To do this, we recall that ${\bf T}_k(p,\alpha)$ in \eqref{eq:T_k} can be expressed as
\begin{align}
  {\bf T}_k(p,\alpha) = \left( {\begin{array}{*{20}{c}}
\left(\beta_k + \frac{M^\gamma}{\tau E_u} \right) {\bf I}_M &\alpha \beta_k {\bf I}_M & \cdots &{{\alpha ^p}} \beta_k {\bf I}_M\\
\alpha \beta_k {\bf I}_M &\left(\beta_k + \frac{M^\gamma}{\tau E_u} \right) {\bf I}_M& \cdots &{{\alpha ^{p - 1}}} \beta_k {\bf I}_M\\
 \vdots & \vdots & \ddots & \vdots \\
{{\alpha ^p}} \beta_k {\bf I}_M&{{\alpha ^{p - 1}}} \beta_k {\bf I}_M& \cdots &\left(\beta_k + \frac{M^\gamma}{\tau E_u} \right) {\bf I}_M
\end{array}} \right).
\end{align}

A close observation shows that, when $M \rightarrow \infty$, the off-diagonal elements of ${\bf T}_k(p,\alpha)$ become negligible due to the fact that $\frac{M^\gamma}{\tau E_u} \gg \alpha^i \beta_k$ $(i = 1, 2, \ldots, p)$. As such, the inverse of ${\bf T}_k(p,\alpha)$ can be accurately approximated by
\begin{align}
 {\bf T}^{-1}_k(p,\alpha) - \left(\frac{M^\gamma}{\tau E_u}\right)^{-1} {\bf I}_{M(p+1)} \overset{M \rightarrow \infty}{\longrightarrow} 0.
\end{align}

Hence, we have
\begin{align}\label{eq:limit_theta}
{\bf \Theta}_k(p,\alpha) - \beta_k^2 \frac{\tau E_u}{M^\gamma} \sum\limits_{j = 0}^p{\alpha^{2j}} {\bf I}_{M} \overset{M \rightarrow \infty}{\longrightarrow} 0.
\end{align}

As a result, \eqref{eq:mrc_g_limit} can be further simplified to
\begin{align}\label{eq:mrc_g_limit_simple}
 \frac{1}{M} {\breve{\bf g}}^\dag_k[n+1]{\breve{\bf g}}_i[n+1] -
  \begin{cases}
   \alpha^2 \beta_k^2 \frac{\tau E_u}{M^\gamma} \sum\limits_{j = 0}^p{\alpha^{2j}}, & i =k\\
   0, & i\neq k
  \end{cases}
  \overset{M \rightarrow \infty}{\longrightarrow} 0.
\end{align}
and
\begin{align}\label{eq:limit_theta_simple}
  \frac{E_u}{M^\gamma} \sum\limits_{i = 1}^K {\frac{1}{M}{\text {tr}}\left( \beta_i{\bf I}_M - \alpha^2 {\bf \Theta}_i(p,\alpha)\right)}
 + 1 =   \frac{E_u}{M^\gamma} \sum\limits_{i = 1}^K \left( \beta_i - \alpha^2 \beta_i^2 \frac{\tau E_u}{M^\gamma} \sum\limits_{j = 0}^p{\alpha^{2j}} \right)
 + 1 \rightarrow 1.
\end{align}

To this end, substitution of \eqref{eq:mrc_g_limit_simple} and \eqref{eq:limit_theta_simple} into \eqref{eq:R_p_k_mrc}, and then combination with \eqref{eq:R_p_k} concludes the proof.
\endproof

Compared to Theorem \ref{theo:MRC_law}, Theorem \ref{theo:MRC_law_pre} indicates that the channel prediction does not alter the power scaling law. Hence, without degradation of the achievable rate, the transmit power of each user can be cut down at most by $1/\sqrt{M}$. As such, setting $\gamma = \frac{1}{2}$, we have the following result.

\begin{corollary}\label{coro:MRC_predicted_CSI}
For MRC receivers, with \underline{p}redicted CSI, each user can scale down its transmit power at most by $p_u = E_u/\sqrt{M}$ for a fixed $E_u$, and the achievable uplink rate of the \emph{k}-th user becomes
  \begin{align}
    R_k^{\text{p,mrc}} \rightarrow \log_2 \left(1+ {\alpha^2 \sum\limits_{j = 0}^p{\alpha^{2j}} \tau E_u^2 \beta_k^2} \right), M \rightarrow \infty.
  \end{align}
\end{corollary}

Corollary \ref{coro:MRC_predicted_CSI} indicates that, although channel prediction does not affect the power scaling law, it does increase the achievable rate by contributing to the enhancement of the effective SINR due to the more accurate CSI being obtained compared to the system without channel prediction. Moreover, the higher the prediction order $p$, the larger the achievable rate gain. However, it is also worth pointing out that the processing complexity increases substantially when the prediction order $p$ becomes large. As such, one should carefully balance this during the system design. For sufficiently large $p$, the achievable rate $R_k^{\text{p,mrc}}$ converges to $\log_2 \left(1+  \frac{\alpha^2}{1 - \alpha^2} \tau E_u^2 \beta_k^2 \right)$, which indicates that the rate gain due to channel prediction is most pronounced for large $\alpha$ and becomes negligible for small $\alpha$. This is rather intuitive, since large $\alpha$ implies relatively slow change of the channel, as such, the channel prediction becomes more accurate.

We now concentrate on the finite $M$ regime, and present the following tight lower bound on the achievable rate of the \emph{k}-th user.
\begin{theorem}\label{theo:MRC_bound_pre}
  For MRC receivers, with \underline{p}redicted CSI, the achievable uplink rate of the \emph{k}-th user is lower bounded by ${R}_k^\text{p,mrc} \geq {\tilde R}_k^\text{p,mrc}$ with
  \begin{align}\label{eq:R_bound_mrc_pre}
    {\tilde R}_k^\text{p,mrc} \triangleq \log_2\left( 1+ \frac{p_u(M-1)\frac{1}{M}{\text{tr}}(\alpha^2{\bf \Theta}_k(p,\alpha))}{p_u\sum\limits_{i = 1, i\neq k}^K{\frac{1}{M}{\text {tr}}\left( \alpha^2 {\bf \Theta}_i(p,\alpha)\right)} + p_u\sum\limits_{i = 1}^K{\frac{1}{M}{\text {tr}}\left( \beta_i{\bf I}_M - \alpha^2 {\bf \Theta}_i(p,\alpha)\right)} + 1 } \right).
  \end{align}
\end{theorem}
\proof The proof follows similar lines as the proof of Theorem \ref{theo:MRC_bound}. Hence, it is omitted. \endproof

%
\subsection{ZF Receivers}
\begin{theorem}\label{theo:ZF_law_pre}
   For ZF receivers, with \underline{p}redicted CSI, if each user scales down its transmit power proportionally to $1/M^\gamma$, i.e., $p_u = E_u/M^\gamma$, where $\gamma >0$ and $E_u$ is fixed, we have
  \begin{align}
R_k^{\text{p,zf}} - \log_2 \left(1+ \frac{\alpha^2 \sum\limits_{j = 0}^p{\alpha^{2j}} \tau E_u^2 \beta_k^2}{M^{2\gamma-1}} \right) \overset{M \rightarrow \infty}{\longrightarrow} 0.
  \end{align}
\end{theorem}
\begin{proof}
With ZF receivers, ${\hat{\bf A}}^\dag[n+1] = ({\breve{\bf G}}^\dag[n+1] {\breve{\bf G}}[n+1])^{-1}{\breve{\bf G}}^\dag[n+1]$, or ${\hat{\bf A}}^\dag[n+1]{\breve{\bf G}}[n+1] = {\bf I}_K$, where ${\breve{\bf G}}[n+1] = [{\breve {\bf g}}_1[n+1], {\breve {\bf g}}_2[n+1], \cdots, {\breve {\bf g}}_K[n+1]]$. Based on this, substituting $p_u = E_u/M^\gamma$ into \eqref{eq:R_p_k}, and after some simple manipulations, we get
\begin{align}\label{eq:R_p_k_zf_2}
 R^\text{p,zf}_k = {\rm E}\left\{ { \log_2 \left(1+ \frac{M\frac{E_u}{M^\gamma}}{ \left( \frac{E_u}{M^\gamma} \sum\limits_{i = 1}^K {\frac{1}{M}{\text {tr}}\left( \beta_i{\bf I}_M - \alpha^2 {\bf \Theta}_i(p,\alpha)\right)}
 + 1 \right)\left[ \left( \frac{\breve{\bf G}^\dag[n+1]{\breve{\bf G}}[n+1]}{M} \right)^{-1}\right]_{k,k}} \right)} \right\}.
\end{align}

Then, we have
\begin{align}\label{eq:G_breve}
\left[\left(\frac{{\breve{\bf G}}^\dag[n+1]{\breve{\bf G}}[n+1]}{M}\right)^{-1}\right]_{k,k} - \frac{1}{\alpha^2 \beta_k^2 \frac{\tau E_u}{M^\gamma} \sum\limits_{j = 0}^p{\alpha^{2j}} } \overset{M \rightarrow \infty}{\longrightarrow} 0,
\end{align}
where the above result is obtained by first following from the law of large numbers and then being based on \eqref{eq:limit_theta}. Plugging \eqref{eq:G_breve} and \eqref{eq:limit_theta_simple} into \eqref{eq:R_p_k_zf_2}, and after some simple algebraic manipulations, we obtain the desired result.
\end{proof}

As expected, Theorem \ref{theo:ZF_law_pre} indicates that with predicted CSI, ZF receivers achieve the same asymptotic power scaling law as the MRC receivers. Similarly, by setting $\gamma = \frac{1}{2}$, we have the following corollary.
\begin{corollary}\label{coro:ZF_predicted_CSI}
For ZF receivers, with \underline{p}redicted CSI, each user can scale down its transmit power at most by $p_u = E_u/\sqrt{M}$ for a fixed $E_u$, and the achievable uplink rate of the \emph{k}-th user becomes
  \begin{align}
 R_k^{\text{p,zf}} \rightarrow \log_2 \left(1+ {\alpha^2 \sum\limits_{j = 0}^p{\alpha^{2j}} \tau E_u^2 \beta_k^2} \right), M \rightarrow \infty.
  \end{align}
\end{corollary}

In the finite $M$ regime, we obtain the following lower bound on the achievable rate.
\begin{theorem}\label{theo:ZF_bound_pre}
  For ZF receivers, with \underline{p}redicted CSI, the achievable uplink rate of the \emph{k}-th user is lower bounded by ${R}_k^\text{p,zf} \geq {\tilde R}_k^\text{p,zf}$ with
  \begin{align}\label{eq:R_bound_zf_pre}
    {\tilde R}_k^\text{p,zf} \triangleq \log_2\left( 1+ \frac{p_u(M-K)\frac{1}{M}{\text{tr}}(\alpha^2{\bf \Theta}_k(p,\alpha))}{ p_u\sum\limits_{i = 1}^K{\frac{1}{M}{\text {tr}}\left( \beta_i{\bf I}_M - \alpha^2 {\bf \Theta}_i(p,\alpha)\right)} + 1 } \right).
  \end{align}
\end{theorem}
\proof Since the proof follows similar lines as the proof of Theorem \ref{theo:ZF_bound}, it is omitted. \endproof

\section{Extension to Single-cell Downlink}
We now extend the analysis to the single-cell downlink scenario. For exposition purpose, only the MRT precoding scheme is considered here. For the single-cell downlink communication, the BS broadcasts data to the $K$ users. Hence, the received signal at user $k$ for the $(n+1)$-th symbol can be expressed as
    \begin{align}\label{eq:dl:y}
      &y_k^\text{dl}[n+1] \notag\\
      &= \sqrt{p_b} {\bf g}_k^T[n+1] {\bf W}[n+1] {\bf x}^\text{dl}[n+1] + z_k^\text{dl}[n+1]\\
      &= \sqrt{p_b} {\bf g}_k^T[n+1] {\bf w}_k[n+1] {x}_k^\text{dl}[n+1] + \sqrt{p_b} \sum_{i\neq k}{\bf g}_k^T[n+1] {\bf w}_i[n+1] {x}_i^\text{dl}[n+1] + z_k^\text{dl}[n+1],\notag
    \end{align}
    where ${\bf x}^\text{dl}[n+1]$ is a $K \times 1$ vector consisting of the transmit symbols to $K$ users with unit power with ${x}_k^\text{dl}[n+1]$ being the \emph{k}-th element of ${\bf x}^\text{dl}[n+1]$; $z_k^\text{dl}[n+1]$ represents the zero-mean additive white Gaussian noise with unit variance; $p_b$ is the transmit power of the BS; ${\bf W}[n+1] \in {\cal{C}}^{M \times K}$ denotes the precoding matrix, and ${\bf w}_k[n+1]$ is the \emph{k}-the vector of the matrix ${\bf W}[n+1]$.

    For the MRT precoding scheme, the beamforming matrix ${\bf W}[n+1]$ is given by
    \begin{align}
      {\bf W}[n+1] = \lambda {\bar{ \bf G}}^*[n+1],
    \end{align}
    where the normalization constant $\lambda$ is chosen to satisfy a long-term total transmit power constraint at the BS, i.e., ${\rm E} \left\{ ||{\bf W}[n+1] {\bf x}^\text{dl}[n+1] ||^2 \right\} = 1$, and we have
    \begin{align}\label{eq:dl:beta}
      \lambda = \sqrt{\frac{1}{M \alpha^2 \sum_{k=1}^K \sigma_k^2 }},
    \end{align}
   where we set $\sigma_k^2 = \frac{p_p\beta_k^2}{1 + p_p\beta_k}$ for notational simplicity.

    Based on the above analysis, we can rewrite \eqref{eq:dl:y} as
    \begin{align}\label{eq:dl:y:re1}
      y_k^\text{dl}[n+1] &= \sqrt{p_b} \lambda {\bf g}_k^T[n+1] {\bar{\bf g}}_k^*[n+1] {x}_k^\text{dl}[n+1]\notag\\
      & + \sqrt{p_b} \lambda \sum_{i = 1, i\neq k}^K {\bf g}_k^T[n+1] {\bar{\bf g}}_i^*[n+1] {x}_i^\text{dl}[n+1] + z_k^\text{dl}[n+1].
    \end{align}
    To obtain the downlink achievable rate, we utilize the technique developed in [24], which is widely used in the analysis of massive MIMO systems. With this technique, the received signal is rewritten as a known mean gain times the desired symbol, plus an uncorrelated effective noise. Thus \eqref{eq:dl:y:re1} can be re-expressed as
    \begin{align}
      y_k^\text{dl}[n+1] = \sqrt{p_b} \lambda {\rm E} \left\{ {\bf g}_k^T[n+1] {\bar{\bf g}}_k^*[n+1] \right\} {x}_k^\text{dl}[n+1] + n_k[n+1],
    \end{align}
    where $n_k[n+1]$ is considered as the effective noise, given by
    \begin{align}
     n_k[n+1] &= \sqrt{p_b} \lambda \left( {\bf g}_k^T[n+1] {\bar{\bf g}}_k^*[n+1] - {\rm E} \left\{ {\bf g}_k^T[n+1] {\bar{\bf g}}_k^*[n+1] \right\} \right) {x}_k^\text{dl}[n+1] \\
      &+ \sqrt{p_b} \lambda \sum_{i = 1, i\neq k}^K {\bf g}_k^T[n+1] {\bar{\bf g}}_i^*[n+1] {x}_i^\text{dl}[n+1] + z_k^\text{dl}[n+1].
    \end{align}
    Therefore, we can obtain an achievable (sub-optimal) rate as
    \begin{align}\label{eq:dl:R}
      R_k^\text{dl} = \log_2 \left(1 + \frac{|{\rm E} \left\{ {\bf g}_k^T[n+1] {\bar{\bf g}}_k^*[n+1] \right\}|^2}{{\text{Var}} \left({\bf g}_k^T[n+1] {\bar{\bf g}}_k^*[n+1]\right) + \sum_{i=1,i\neq k}^K {\rm E} \left\{ |{\bf g}_k^T[n+1] {\bar{\bf g}}_i^*[n+1]|^2 \right\} + \frac{1}{p_b\lambda^2} } \right).
    \end{align}
    \begin{theorem}
    For MRC receivers, with \underline{a}ged CSI, the achievable downlink rate of the \emph{k}-th user is given by
      \begin{align}\label{eqn:11}
      R_k^\text{dl} = \log_2 \left(1 + \frac{\alpha^2 M \sigma_k^4}{\left(\beta_k + \frac{1}{p_b}\right) \sum_{i=1}^K \sigma_i^2 } \right).
    \end{align}
    \end{theorem}

    \proof
    The main task is to evaluate each term in \eqref{eq:dl:R}, which we do in the following:

    1) Computation of ${\rm E} \left\{ {\bf g}_k^T[n+1] {\bar{\bf g}}_k^*[n+1] \right\}$

    We have
    \begin{align}\label{eq:dl:g}
      {\bf g}_k^T[n+1] {\bar{\bf g}}_k^*[n+1] &= {\bar{\bf g}}_k^T[n+1] {\bar{\bf g}}_k^*[n+1] + {\tilde{\bf e}}_k^T[n+1] {\bar{\bf g}}_k^*[n+1]\\
       &= \alpha^2 ||{\hat{\bf g}}_k^T[n] ||^2 + {\tilde{\bf e}}_k^T[n+1] {\bar{\bf g}}_k^*[n+1].
    \end{align}
    Therefore,
    \begin{align}\label{eq:dl:g:E}
      {\rm E} \left\{ {\bf g}_k^T[n+1] {\bar{\bf g}}_k^*[n+1] \right\} = \alpha^2 {\rm E} \left\{ ||{\hat{\bf g}}_k[n]||^2 \right\} = \alpha^2 M \sigma_k^2.
    \end{align}
    2) Computation of ${\text{Var}} \left({\bf g}_k^T[n+1] {\bar{\bf g}}_k^*[n+1]\right)$

    From \eqref{eq:dl:g} and \eqref{eq:dl:g:E}, the variance of ${\bf g}_k^T[n+1] {\bar{\bf g}}_k^*[n+1]$ is given by
    \begin{align}
      &{\text{Var}} \left({\bf g}_k^T[n+1] {\bar{\bf g}}_k^*[n+1]\right) = {\rm E} \left\{ |{\bf g}_k^T[n+1] {\bar{\bf g}}_k^*[n+1]|^2 \right\} - \left(\alpha^2 M \sigma_k^2 \right)^2 \\
      &= {\rm E} \left\{ | \alpha^2 ||{\hat{\bf g}}_k^T[n] ||^2 + {\tilde{\bf e}}_k^T[n+1] {\bar{\bf g}}_k^*[n+1] |^2\right\} - \left(\alpha^2 M \sigma_k^2 \right)^2 \\
      &= \alpha^4 {\rm E} \left\{ ||{\hat{\bf g}}_k^T[n] ||^4 \right\} + {\rm E} \left\{ |{\tilde{\bf e}}_k^T[n+1] {\bar{\bf g}}_k^*[n+1]|^2 \right\} - \left(\alpha^2 M \sigma_k^2 \right)^2.
    \end{align}
    By using \cite[Lemma 2.9]{A.M.Tulino}, we obtain
    \begin{align}\label{eq:dl:var}
      {\text{Var}} \left({\bf g}_k^T[n+1] {\bar{\bf g}}_k^*[n+1]\right) = \alpha^4 \sigma_k^4 M(M+1) + \alpha^2 \sigma_k^2 (\beta_k - \alpha^2\sigma_k^2)M - \left(\alpha^2 M \sigma_k^2 \right)^2 = \alpha^2\sigma_k^2\beta_k M.
    \end{align}
    3) Computation of $\sum_{i=1,i\neq k}^K {\rm E} \left\{ |{\bf g}_k^T[n+1] {\bar{\bf g}}_i^*[n+1]|^2 \right\}$

    For $i\neq k$, we have
    \begin{align}
      {\rm E} \left\{ |{\bf g}_k^T[n+1] {\bar{\bf g}}_i^*[n+1]|^2 \right\} = \alpha^2 \beta_k \sigma_i^2 M.
    \end{align}
    Therefore,
    \begin{align}\label{eq:dl:sum}
      \sum_{i=1,i\neq k}^K {\rm E} \left\{ |{\bf g}_k^T[n+1] {\bar{\bf g}}_i^*[n+1]|^2 \right\} = \alpha^2 \beta_k M \sum_{i=1,i\neq k}^K \sigma_i^2.
    \end{align}

    Substituting \eqref{eq:dl:beta}, \eqref{eq:dl:g:E}, \eqref{eq:dl:var}, and \eqref{eq:dl:sum} into \eqref{eq:dl:R}, we arrive at the desired result.
    \endproof

    \begin{theorem}
    For MRC receivers, with \underline{a}ged CSI, if the BS scales down its transmit power proportionally to $1/M^\beta$, i.e., $p_b = E_b/M^\beta$, where $\beta > 0$ and $E_b$ is fixed, we have
  \begin{align}\label{eqn:12}
 R_k^\text{dl} - \log_2 \left(1 + \frac{\alpha^2 {\tau E_u E_b} \beta_k^4 }{{M^{\beta - \frac{1}{2}}} \sum_{i=1}^K \beta_i^2} \right) \overset{M \rightarrow \infty}{\longrightarrow} 0.
  \end{align}
    \end{theorem}
  \proof
  As in the uplink scenario, substituting $p_p = \tau p_u = \tau \frac{E_u}{\sqrt{M}}$ into \eqref{eqn:11}, and after some simple algebraic manipulations, we obtain the desired result.
  \endproof
  When $\beta = \frac{1}{2}$, $R_k^\text{dl}$ converges to a non-zero limit, indicating that we can at most scale down the transmit power of the BS proportionally to $1/\sqrt{M}$ in the downlink scenario, which is the same as in the uplink scenario.

\section{Extension to Multi-Cell Systems}
In this section, we study the more general multi-cell scenario. In particular, we focus on the characterization of the power scaling law of the system with/without channel prediction. Since MRC and ZF receivers attain the same asymptotic performance, without loss of generality, we only consider the MRC receiver in the subsequent analysis.

We adopt the multi-cell model as in \cite{K.T.Truong1}, with a cellular network of $C$ cells sharing the same frequency band. Each cell includes a central BS with $M$ antennas and $K$ $(K \leq M)$ single antenna noncooperative users. Therefore, the $M \times 1$ received vector at time $n$ for the \emph{b}-th BS is given by
\begin{align}
  {\bf y}_b[n] = \sqrt{p_u} \sum\limits_{c = 1}^C  {\bf G}_{bc}[n] {\bf x}_c[n] + {\bf z}_b[n],
\end{align}
where ${\bf G}_{bc}[n]$ represents the $M \times K$ matrix between the \emph{b}-th BS and the $K$ users in the \emph{c}-th cell, whose \emph{k}-th column vector is denoted by ${\bf g}_{bck}[n]$, $p_u$ is the transmit power of the user, ${\bf x}_c$ denotes the $K \times 1$ vector transmitted by the $K$ users in the \emph{c}-th cell, and ${\bf z}_b$ is zero-mean additive white Gaussian noise with unit power at BS $b$.

Similar to the single-cell scenario, the channel vector from user $k$ in cell $c$ to BS $b$ at time $n$ is modeled as
\begin{align}
  {\bf g}_{bck}[n] =  {\bf h}_{bck}[n] \sqrt{\beta_{bck}},
\end{align}
where ${\bf h}_{bck}[n]$ is the small-scale fading coefficient from the \emph{k}-th user in cell $c$ to the \emph{b}-th BS, which is i.i.d. ${{\cal CN} (0,1)}$, and $\beta_{bck}$ models the large-scale fading effect.

Capitalizing on the asymptotic expressions $(53)$ and $(75)$ presented in \cite{K.T.Truong1}, we obtain the following results on the power scaling law for MRC receivers.
\begin{proposition}\label{prop:mul:cell:aged:CSI}
  For the multi-cell scenario, with \underline{a}ged CSI, if each user scales down its transmit power proportionally to $1/M^\gamma$, i.e., $p_u = E_u/M^\gamma$, where $\gamma >0$ and $E_u$ is fixed, the achievable rate of user $k$ in cell $b$ is given by
\begin{align}\label{eq:mul:cell:aged}
    R^\text a_{bk} - \log_2 \left( 1 + \frac{ \frac{\alpha^2 \tau E_u^2 \beta^2_{bbk}}{M^{2\gamma - 1}}}{\frac{\beta_{bbk} E_u}{M^\gamma} + 1 + \sum\limits_{(c,i) \neq (b,k)} \frac{\beta_{bci} E_u}{M^\gamma} + \alpha^2 \sum\limits_{ c \neq b} \frac{\tau E_u^2 \beta^2_{bck}}{M^{2 \gamma - 1}} } \right) \overset{M \rightarrow \infty}{\longrightarrow} 0,
  \end{align}
 where the third term in the denominator is due to both the intra and the inter-cell interference, while the fourth term comes from the inter-cell interference caused by pilot contamination.
\end{proposition}
\proof Substituting $p_u = E_u/M^\gamma$ into Theorem 2 of \cite{K.T.Truong1}, the desired result can be obtained after some lengthy algebraic manipulations.\endproof

As expected, Proposition \ref{prop:mul:cell:aged:CSI} suggests that the asymptotic achievable rate $R_{bk}$ depends on the choice of $\gamma$. It is easy to show that for $\gamma >1/2$, the achievable rate $R_{bk}$ converges to zero. Similarly, for $\gamma = 1/2$, $R_{bk}$ converges to a non-zero limit given by
\begin{align}
  R^\text a_{bk} \rightarrow \log_2 \left( 1 + \frac{ \alpha^2 \tau E_u^2 \beta^2_{bbk} }{1 + \alpha^2 \sum\limits_{c \neq b}  \tau E_u^2 \beta^2_{bck} } \right), M \rightarrow \infty.
  \end{align}

Once again, we see that the $1/\sqrt{M}$ power scaling law still holds under the multi-cell scenario. In addition, it is observed that the non-zero limit is affected not only by the channel aging effect, but also by the inter-cell interference due to the pilot contamination caused by pilot reuse.

We now look at the case $0 < \gamma<1/2$. Interestingly, it is found that, unlike the single-cell scenario where the achievable rate grows unbounded, $R_{bk}$ also converges to a non-zero limit given by
  \begin{align}\label{eqn:qbk}
  R^\text a_{bk} \rightarrow \log_2 \left( 1 + \frac{\beta^2_{bbk} }{\sum\limits_{ c \neq b}  \beta^2_{bck} } \right), M \rightarrow \infty.
  \end{align}

Surprisingly, we see that the effect of channel aging vanishes and $R_{bk}$ is independent of the transmit power which coincides with the results presented in \cite{T.L.Marzetta}. The possible reason is that, when $0 < \gamma <1/2$, the system operates in an interference-limited regime; as such, if each terminal scales its average received power by the same factor, then the resultant signal-to-interference ratio (SIR) remains unchanged.

We now consider the case with channel prediction, and present the following key result:
\begin{proposition}\label{prop:mul:cell:predicted:CSI}
   For the multi-cell scenario, with \underline{p}redicted CSI, if each user scales down its transmit power proportionally to $1/M^\gamma$, i.e., $p_u = E_u/M^\gamma$, where $\gamma >0$ and  $E_u$ is fixed, the achievable rate of user $k$ in cell $b$ is given by
 \begin{align}\label{eq:mul:cell:predicted}
    R^\text p_{bk} - \log_2 \left( 1 + \frac{ \frac{\alpha^2 \sum\limits_{j = 0}^p \alpha^{2j} \tau E_u^2 \beta^2_{bbk}}{M^{2\gamma - 1}}}{\frac{\beta_{bbk} E_u}{M^\gamma} + 1 + \sum\limits_{(c,i) \neq (b,k)} \frac{\beta_{bci} E_u}{M^\gamma} + \alpha^2 \sum\limits_{j = 0}^p \alpha^{2j} \sum\limits_{ c \neq b} \frac{\alpha^{2p} \tau E_u^2 \beta^2_{bck}}{M^{2 \gamma - 1}} } \right) \overset{M \rightarrow \infty}{\longrightarrow} 0.
  \end{align}
\end{proposition}
\proof Substituting $p_u = E_u/M^\gamma$ into the Theorem 3 of \cite{K.T.Truong1}, and after some tedious algebraic manipulations, we get the desired result.\endproof

We now discuss the impact of $\gamma$ on the asymptotic achievable rate based on Proposition \ref{prop:mul:cell:predicted:CSI}. It can be easily shown that, for $\gamma>1/2$, $R_{bk} \rightarrow 0$, and for $\gamma=1/2$,
  \begin{align}
  R^\text p_{bk} \rightarrow \log_2 \left( 1 + \frac{ \alpha^2 \sum\limits_{j = 0}^p \alpha^{2j} \tau E_u^2 \beta^2_{bbk} }{1 + \alpha^2 \sum\limits_{j = 0}^p \alpha^{2j} \sum\limits_{c \neq b} \alpha^{2p} \tau E_u^2 \beta^2_{bck} } \right), M \rightarrow \infty.
  \end{align}

As expected, the asymptotic achievable rate is determined by both the channel aging effect and the inter-cell interference. Similarly, for $\gamma<1/2$, we have
  \begin{align}
  R^\text p_{bk} \rightarrow \log_2 \left( 1 + \frac{ \beta^2_{bbk} }{\alpha^{2p}\sum\limits_{c \neq b}  \beta^2_{bck} } \right), M \rightarrow \infty.
  \end{align}

Now, compared to the achievable rate of systems with aged CSI presented in (\ref{eqn:qbk}), we observe that the achievable rate with channel prediction is strictly higher, due to the fact that $\alpha^{2p}<1$.

\section{Numerical Results}
In this section, we provide numerical results to validate the analytical expressions derived in the previous sections. Hereafter, we assume that $\tau = K$.
\subsection{Single-cell uplink scenario}
We consider a single-cell with a radius of $R = 1000$ meters and assume a guard range of $r_0 = 100$ meters, which specifies the distance between the nearest user and the BS. All the users are uniformly distributed within the cell. The large scale fading is modeled as $\beta_k = z_k/(r_k/r_0)^\upsilon$, where $z_k$ is a log-normal random variable with standard deviation $\sigma$ ($\sigma = 8$ dB) denoting the shadow fading effect, $r_k$ represents the distance between the \emph{k}-th user and the BS, and $\upsilon$ ($\upsilon = 3.8$) is the path loss exponent.

Fig. \ref{fig:rate_snr} examines the tightness of the proposed analytical lower bounds given in \eqref{eq:R_bound_mrc} and \eqref{eq:R_bound_zf} with aged CSI, as well as in \eqref{eq:R_bound_mrc_pre} and \eqref{eq:R_bound_zf_pre} with predicted CSI for different $\alpha$ and $p$. As can be readily observed, the proposed lower bounds almost overlap with the exact simulations curves, demonstrating their tightness. In addition, we see the intuitive result that channel aging degrades the achievable sum-rate, while channel prediction helps to recover part of the sum-rate loss due to channel aging. Moreover, it is observed that the ZF receivers attain a higher sum-rate than the MRC receivers. Finally, we observe that channel aging causes a substantial reduction in the sum rate of ZF receivers and a relatively small decrease in the sum rate of MRC receivers at the high SNR regime, indicating that the channel aging effect has a much greater impact on ZF receivers. This can be attributed to the poor interference cancellation capabilities of ZF receivers when channel aging is present, while MRC tries to maximize the effective SINR of each target user.

\begin{figure}[!ht]
    \centering
    \includegraphics[scale=0.6]{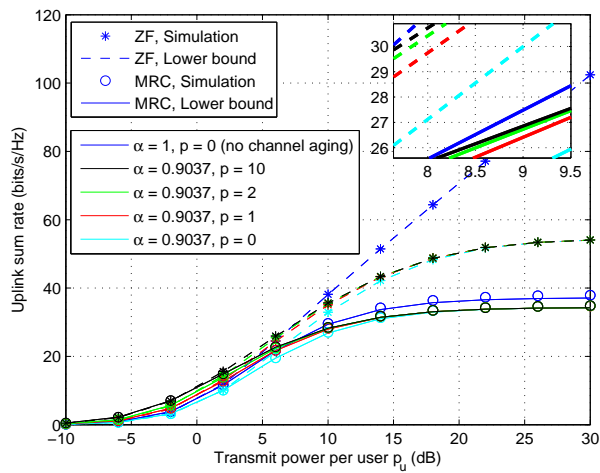}
    \caption{Uplink sum-rate versus the transmit power $p_u$ for $K = 10$, $M = 128$, and $f_DT_s = 0.1$.}\label{fig:rate_snr}
  \end{figure}

Fig. \ref{fig:rate_fD_Ts} investigates the impact of channel prediction on the achievable sum-rate lower bound. Note that the curves associated with perfect CSI are obtained from \cite[(17) and (21)]{H.Q.Ngo}, while the curves associated with current CSI are based on \eqref{eq:R_bound_mrc} and \eqref{eq:R_bound_zf} by setting $\alpha = 1$. Intuitively, as the normalized Doppler shift $f_DT_s$ becomes large, i.e., for stronger channel aging effect, the sum-rate loss becomes increasingly substantial. Also, the higher the prediction order, the larger the sum-rate gain. In addition, the benefit of channel prediction tends to be more significant when the channel aging effect is less severe, i.e., $f_DT_s$ is small. This is rather expected, since the predicted CSI becomes more accurate in such scenarios. Finally, it is observed that the predicted CSI case achieves a higher rate than the current CSI case when $f_DT_s$ is small, while its performance degrades substantially when $f_DT_s$ is large, and becomes worse than that with current CSI case. This can be explained as follows: When the channel varies slowly, channel prediction which uses multiple channel observations can provide more accurate CSI than channel estimation which only uses one channel observation. On the other hand, when the channel varies fast, channel prediction becomes inaccurate, and is less reliable than the current CSI. However, it is also worth pointing out that the achievable rate with channel prediction can not exceed the rate achieved with perfect CSI.

  \begin{figure}[!ht]
    \centering
    \includegraphics[scale=0.6]{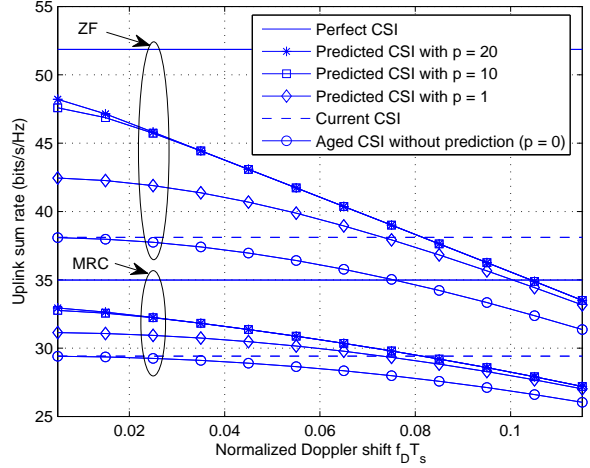}
    \caption{Lower bound on the uplink sum-rate versus the normalized Doppler shift $f_DT_s$ for $K = 10$, $M = 128$, and $p_u = 10$ dB.}\label{fig:rate_fD_Ts}
  \end{figure}

Fig. \ref{fig:rate_m_scaling} shows the lower bound on the achievable sum-rate versus the number of BS antennas when the transmit power of each user is scaled down by $1/\sqrt{M}$. As predicted by Corollary \ref{coro:MRC_aged_CSI} and \ref{coro:MRC_predicted_CSI}, the achievable sum-rate converges to a non-zero limit when the number of antennas $M$ becomes large. As the prediction order increases, the sum-rate improves. Nevertheless, we also observe that, the sum-rate gain due to increasing the prediction order from $p = 1$ to $p = 2$ is significantly smaller than the sum-rate gain from increasing $p = 0$ to $p = 1$. Recall from Corollary \ref{coro:MRC_predicted_CSI}, that the gain due to channel prediction is manifested through the SNR enhancement factor $\sum\limits_{j = 1}^p{\alpha^{2j}}$. When $\alpha$ is relatively small, the contribution of a higher $p$ diminishes quickly, which explains the above behavior.

\begin{figure}[!ht]
    \centering
    \includegraphics[scale=0.6]{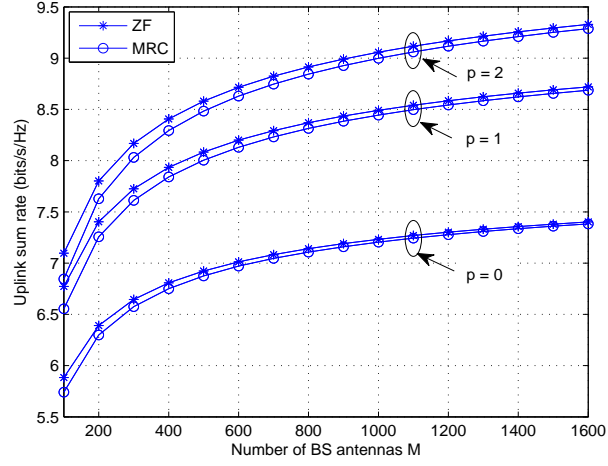}
    \caption{Lower bound on the uplink sum-rate versus the number of BS antennas $M$ for $K = 10$, $f_DT_s = 0.1$, $p_u = E_u/\sqrt{M}$ with $E_u = 15$ dB.}\label{fig:rate_m_scaling}
  \end{figure}

\subsection{Single-cell downlink scenario}
Fig. \ref{fig:downlink_rate_snr} verifies the correctness of the analytical expression given in (\ref{eqn:11}). As we can readily observe, the analytical results are in perfect agreement with the simulation curves.
\begin{figure}[ht]
    \centering
    \includegraphics[scale=0.6]{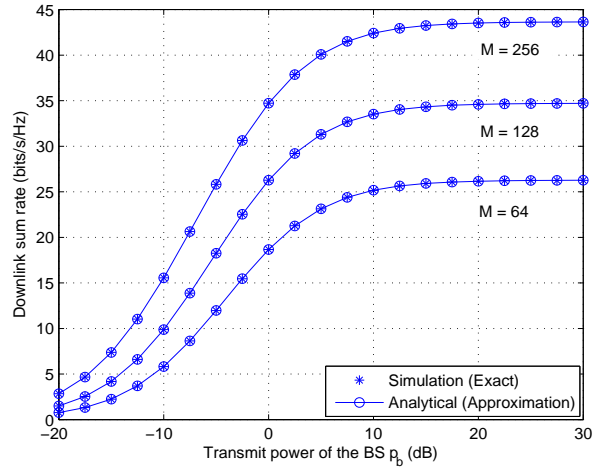}
    \caption{Downlink sum-rate with MRT precoder versus the transmit power $p_b$ for $K = 10$, $M = 64$, $p_p = 10$ dB, and $f_D T_s = 0.1$.}\label{fig:downlink_rate_snr}
\end{figure}

Fig. \ref{fig:downlink_rate_m} illustrates the power scaling law. When $M$ grows large, the analytical results converge to the asymptote. Also, the speed of convergence depends on the transmit power, the smaller the transmit power, the faster the convergence speed. In addition, as can be seen, in the downlink scenario, the power scaling law is also $1/\sqrt{M}$, and identical to the uplink scenario.
\begin{figure}[ht]
    \centering
    \includegraphics[scale=0.6]{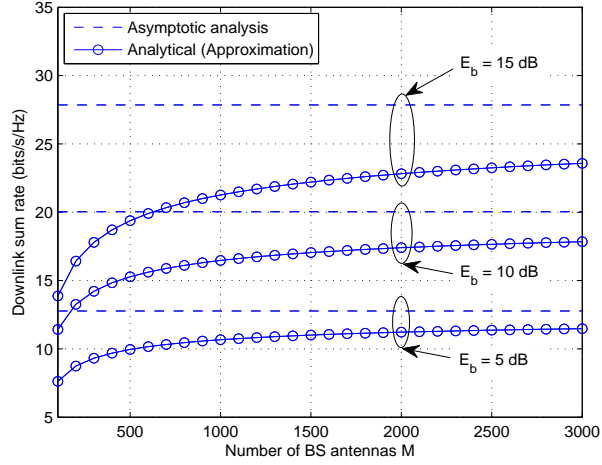}
    \caption{Downlink sum-rate versus the number of BS antennas $M$ with MRT precoder for $K = 5$, $f_DT_s = 0.1$, $p_p = \tau E_u/\sqrt{M}$ with $E_u = 3$ dB ,and $p_b = E_b/\sqrt{M}$.}\label{fig:downlink_rate_m}
\end{figure}
\subsection{Multi-cell scenario}
In this section, we examine the impact of channel aging and channel prediction on the achievable sum-rate of cellular massive MIMO systems. As in \cite{H.Q.Ngo}, we assume $\beta_{bbk} = 1$, and $\beta_{bck} = 0.32$ $(c \neq b)$ for all $k$ $(1 \leq k \leq K)$, and consider a setting with $C = 7$ cells.

    \begin{figure}[ht]
    \centering
    \includegraphics[scale=0.6]{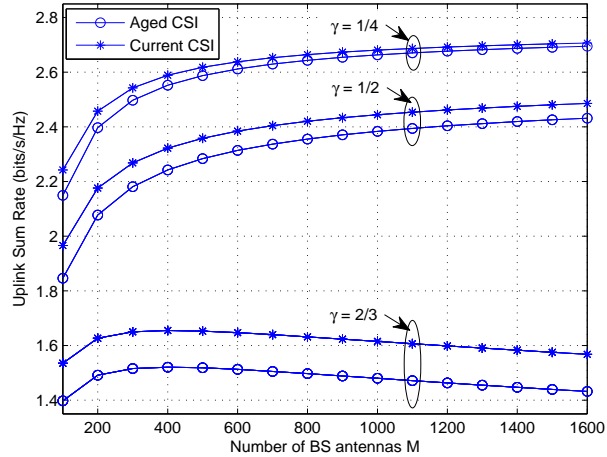}
    \caption{Uplink sum-rate versus the number of BS antennas $M$ for MRC receivers with aged CSI for $K = 10$, $f_DT_s = 0.1$, and $p_u = E_u/M^\gamma$ with $E_u = 15$ dB.}\label{fig:rate_multicell_m_scaling1}
  \end{figure}

Fig. \ref{fig:rate_multicell_m_scaling1} illustrates the power scaling law of multi-cell massive MIMO systems with aged CSI. As expected, we see that when $\gamma>1/2$, the achievable sum-rate gradually decreases, and eventually reduces to zero as $M$ approaches infinity. While for $\gamma\leq 1/2$, the achievable sum-rate converges to a deterministic non-zero value. In addition, when $\gamma=1/2$, we see that, regardless of the antenna number $M$, there exists a constant gap between the two curves associated to scenarios with aged CSI and current CSI, respectively, elucidating the detrimental effect of channel aging. On the other hand, when $\gamma<1/2$, the two curves overlap for sufficiently large $M$, indicating the vanishing effect of channel aging, as indicated by (\ref{eqn:qbk}).

\begin{figure}[ht]
    \centering
    \includegraphics[scale=0.6]{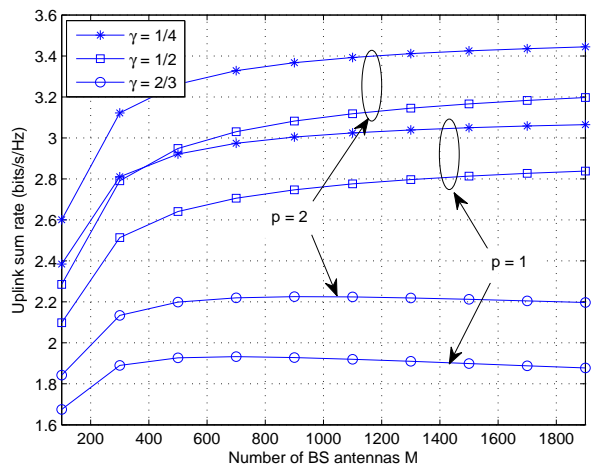}
    \caption{Uplink sum-rate versus the number of BS antennas $M$ for MRC receivers with channel prediction for $K = 10$, $f_DT_s = 0.1$, and $p_u = E_u/M^\gamma$ with $E_u = 15$ dB. }\label{fig:rate_multicell_m_scaling}
\end{figure}

Fig. \ref{fig:rate_multicell_m_scaling} depicts the power scaling law of multi-cell massive MIMO systems with channel prediction, given by \eqref{eq:mul:cell:predicted}. As expected, the achievable sum-rate converges to a non-zero limit when $\gamma\leq 1/2$, and reduces to zero when $\gamma>1/2$ as the number of antennas $M$ increases. Moreover, as the predictor order $p$ increases, the non-zero limit becomes larger.

\section{Conclusion}
This paper studied the achievable sum-rate of uplink massive MIMO systems taking into account the channel aging effect. Specifically, we derived tractable lower bounds of the sum-rate for both MRC and ZF receivers with/without channel prediction, which are valid for arbitrary number of antennas and users. In addition, we characterized the impact of channel aging effect and channel prediction on the power scaling law. The findings of the paper suggest that aged CSI degrades the corresponding achievable sum-rate, and the more severe the channel aging effect, the more significant reduction of the sum-rate. Moreover, channel prediction enhances the sum-rate, and the higher the predictor order, the better the sum-rate performance. In addition, it is shown that the benefits due to channel prediction are more pronounced in the scenario, where the channel aging effect is not severe. Finally, it was found that both in the single-cell and the multi-cell scenario, the transmit power of each user can be scaled down at most by $1/\sqrt{M}$ in the presence of channel aging, which indicates that aged CSI does not degrade the power scaling law, and channel prediction does not improve the power scaling law. Similarly, the single-cell downlink scenario analysis was presented, which concludes that in the single-cell downlink scenario, the same power scaling law $1/\sqrt{M}$ is achieved as in the single-cell uplink scenario. However, unlike the single-cell scenario, the achievable rate in the multi-cell scenario converges to a non-zero limit when each user does not cut down the transmit power by the maximum limit, i.e., $1/M^\gamma$ with $0<\gamma<1/2$, due to the effect of pilot contamination.

\appendices
\section{Proof of Theorem \ref{theo:MRC_bound}}\label{app:theo:MRC_bound}
By substituting ${\hat{\bf a}}_k[n+1] = {\bar{\bf g}}_k[n+1] = \alpha {\hat{\bf g}}_k[n]$ into \eqref{eq:R_k}, we obtain
\begin{align}\label{eq:R_k_mrc}
 R_k^\text{mrc} = {\rm E}\left\{ { \log_2 \left(1+ \frac{p_u\alpha^2||{\hat{\bf g}}_k[n]||^2}{p_u \alpha^2 \sum\limits_{i = 1, i \neq k}^K {|{\tilde g}_i[n]|^2}
 + p_u \sum\limits_{i = 1}^K {\left(\beta_i - \alpha^2\frac{p_p\beta_i^2}{1+p_p\beta_i}\right)}
 + 1} \right)} \right\},
\end{align}
where ${\tilde g}_i[n] \triangleq \frac{{\hat{\bf g}}_k^\dag[n] {\hat{\bf g}}_i[n]}{||{\hat{\bf g}}_k[n]||}$. To this end, noticing that ${\text{log}}_2\left(1+\frac{1}{x}\right)$ is a convex function with respect to $x$, the following tight lower bound can be obtained by applying Jensen's inequality \cite{H.Q.Ngo}
\begin{align}
   {\tilde R}_k^\text{mrc} =  { \log_2 \left(1+ \left({\rm E}\left\{ \frac{p_u \alpha^2 \sum\limits_{i = 1, i \neq k}^K {|{\tilde g}_i[n]|^2}
 + p_u\sum\limits_{i = 1}^K {\left(\beta_i - \alpha^2\frac{p_p\beta_i^2}{1+p_p\beta_i}\right)}
 + 1}{p_u\alpha^2||{\hat{\bf g}}_k[n]||^2} \right\} \right)^{-1} \right)} ,
\end{align}

By noticing that conditioned on ${\hat{\bf g}}_k[n]$, ${\tilde g}_i[n]$ is a Gaussian random variable with zero mean and variance $\frac{p_p\beta_i^2}{1+p_p\beta_i}$, which does not depend on ${\hat{\bf g}}_k[n]$, it is concluded that ${\tilde g}_i[n]$ is Gaussian distributed and independent of ${\hat{\bf g}}_k[n]$, i.e., ${\tilde g}_i[n] \sim {\cal{CN}}\left( 0, \frac{p_p\beta_i^2}{1+p_p\beta_i} \right)$. As a result, we obtain
\begin{align}\label{eq:exp:mrc}
 & {\rm E}\left\{ \frac{p_u \alpha^2 \sum\limits_{i = 1, i \neq k}^K {|{\tilde g}_i[n]|^2}
 + p_u\sum\limits_{i = 1}^K {\left(\beta_i - \alpha^2\frac{p_p\beta_i^2}{1+p_p\beta_i}\right)}
 + 1}{p_u\alpha^2||{\hat{\bf g}}_k[n]||^2} \right\}\notag \\
& = \left( p_u \alpha^2 \sum\limits_{i = 1, i \neq k}^K {\rm E}\left\{ {|{\tilde g}_i[n]|^2}\right\} + p_u\sum\limits_{i = 1}^K {\left(\beta_i - \alpha^2\frac{p_p\beta_i^2}{1+p_p\beta_i}\right)}
 + 1  \right) {\rm E}\left\{ \frac{1}{p_u\alpha^2 ||{\hat{\bf g}}_k[n]||^2}\right\}.
\end{align}

Given that
\begin{align}\label{eq:exp:g_tilde}
  {\rm E}\left\{ {|{\tilde g}_i[n]|^2}\right\} = \frac{p_p\beta_i^2}{1+p_p\beta_i},
\end{align}
the remaining task is to evaluate ${\rm E}\left\{ \frac{1}{||{\hat{\bf g}}_k[n]||^2}\right\}$, which, according to \cite[Lemma 2.10]{A.M.Tulino}, can be computed as
\begin{align}\label{eq:exp:g_breve}
    {\rm E}\left\{ \frac{1}{||{\hat{\bf g}}_k[n]||^2}\right\} = \frac{1}{M-1} \frac{1+p_p\beta_k}{p_p\beta_k^2}.
\end{align}

To this end, after plugging \eqref{eq:exp:g_tilde} and \eqref{eq:exp:g_breve} into \eqref{eq:exp:mrc}, we arrive at the desired result.
\section{Proof of Theorem \ref{theo:ZF_bound}}\label{app:theo:ZF_bound}
With ZF detector, ${\hat{\bf A}}^\dag[n+1] = ({\bar{\bf G}}^\dag[n+1] {\bar{\bf G}}[n+1])^{-1}{\bar{\bf G}}^\dag[n+1]$, or ${\hat{\bf A}}^\dag[n+1]{\bar{\bf G}}[n+1] = {\bf I}_K$. Thus, \eqref{eq:R_k} becomes
\begin{align}
  R_k^\text{zf} = {\rm E}\left\{ { \log_2 \left(1+ \frac{p_u}{\left(p_u\sum\limits_{i = 1}^K {\left(\beta_i - \alpha^2\frac{\tau p_u \beta_i^2}{1+\tau p_u \beta_i}\right)}
 + 1 \right)\left[\left({\bar{\bf G}}^\dag[n+1]{\bar{\bf G}}[n+1]\right)^{-1}\right]_{k,k}} \right)} \right\}.
\end{align}

From \eqref{eq:g_bar_aged}, we have
\begin{align}\label{eq:exp:G_bar}
  {\rm E}\left\{ \left[\left({\bar{\bf G}}^\dag[n+1]{\bar{\bf G}}[n+1]\right)^{-1}\right]_{k,k} \right\} &= \frac{1}{\alpha^2} {\rm E}\left\{ \left[\left({\hat{\bf G}}^\dag[n]{\hat{\bf G}}[n]\right)^{-1}\right]_{k,k} \right\} \\
  &= \frac{1}{\alpha^2(M-K)} \frac{1+p_p\beta_k}{p_p\beta_k^2},
\end{align}
which completes the proof.

 \bibliographystyle{IEEE}

\begin{thebibliography}{10}

\bibitem{J.G.Andrews}
J. G. Andrews, S. Buzzi, Wan Choi, S. V. Hanly, A. Lozano, A. C. K. Soong, and J. C.  Zhang, ``What will 5G be?,'' {\em IEEE J. Sel. Areas Commun}, vol. 32, no. 6, pp. 1065--1082, Jun. 2014.

\bibitem{T.L.Marzetta}
T. L. Marzetta, ``Noncooperative cellular wireless with unlimited numbers of base station antennas,'' {\em IEEE Trans. Wireless Commun.}, vol. 9, no. 11, pp. 3590--3600, Nov. 2010.

\bibitem{E.G.Larsson}
E. G. Larsson, O. Edfors, F. Tufvesson, and T. L. Marzetta, ``Massive MIMO for next generation wireless systems,'' {\em IEEE Commun. Mag.}, vol. 52, no. 2, pp. 186--195, Feb. 2014.

\bibitem{F.Rusek}
F. Rusek, D. Persson, B. K. Lau, E. G. Larsson, T. L. Marzetta, O. Edfors, and F. Tufvesson, ``Scaling up MIMO: Opportunities and challenges with very large arrays,'' {\em IEEE Signal Process. Mag.}, vol. 30, no. 1, pp. 40--60, Jan. 2013.


\bibitem{Q.Zhang}
Q. Zhang, S. Jin, K.-K. Wong, H. Zhu, and M. Matthaiou, ``Power scaling of uplink massive MIMO systems with arbitrary-rank channel means,'' {\em IEEE J. Sel. Topics Signal Process.}, vol. 8, no. 5, pp. 966--981, Oct. 2014.

\bibitem{C.Kong}
C. Kong, C. Zhong, M. Matthaiou, and Z. Zhang, ``Performance of downlink massive MIMO in Ricean fading channels with ZF precoder,'' in {\it Proc.} {\em IEEE ICC}, June 2015.

\bibitem{J.Hoydis}
J. Hoydis, S. ten Brink, and M. Debbah, ``Massive MIMO in the UL/DL of cellular networks: How many antennas do we need?'' {\em IEEE J. Sel. Areas Commun.}, vol. 31, no. 2, pp. 160--171, Feb. 2013.

\bibitem{J.Zhang}
J. Zhang, C-K. Wen, S. Jin, X. Gao, and K-K. Wong, ``On capacity of large-scale MIMO multiple access channels with distributed sets of correlated antennas,'' {\em IEEE J. Sel. Areas Commun.}, vol. 31, no. 2, pp. 133--148, Feb. 2013.

\bibitem{H.Ngo0}
H. Q. Ngo, E. G. Larsson, and T. L. Marzetta, ``The multicell multiuser MIMO uplink with very large antenna arrays and a finite-dimensional channel,'' {\em IEEE Trans. Commun.}, vol. 61, no. 6, pp. 2350--2361, June 2013.

\bibitem{C.Masouros}
C. Masouros, M. Sellathurai, and T. Ratnarajah, ``Large-scale MIMO transmitters in fixed physical spaces: The effect of transmit correlation and mutual coupling,'' {\em IEEE Trans. Commun.}, vol. 61, no. 7, pp. 2794--2804, July 2013.


\bibitem{J.Jose}
J. Jose, A. Ashikhmin, T. L. Marzetta, and S. Vishwanath, ``Pilot contamination and precoding in multi-cell TDD systems'' {\em IEEE Trans. Wireless Commun.}, vol. 10, no. 8, pp. 2640--2651, Aug. 2011.


\bibitem{N.Krishnan2}
N. Krishnan, R. D. Yates, and N. B. Mandayam, ``Uplink linear receivers for multi-cell multiuser MIMO with pilot contamination: Large system analysis,'' {\em IEEE Trans. Wireless Commun.}, vol. 13, no. 8, pp. 4360--4373, Aug. 2014.

\bibitem{J.Choi1}
J. Choi, D. J. Love, and P. Bidigare, ``Downlink training techniques for FDD massive MIMO systems: Open-loop and closed-loop training with memory,'' {\em IEEE J. Sel. Topics Signal Process.}, vol. 8, no. 5, pp. 802-814, Oct. 2014.

\bibitem{S.Noh}
S. Noh, M. Zoltowski, Y. Sung, and D. Love, ``Pilot beam pattern design for channel estimation in massive MIMO systems,'' {\em IEEE J. Sel. Topics Signal Process.}, vol. 8, no. 5, pp. 787 -- 801, Oct. 2014.

\bibitem{H.Q.Ngo}
H. Q. Ngo, E. G. Larsson, and T. L. Marzetta, ``Energy and spectral efficiency of very large multiuser MIMO systems,'' {\em IEEE Trans. Commun.}, vol. 61, no. 4, pp. 1436--1449, Apr. 2013.

\bibitem{C.K.Wen}
C.-K. Wen, S. Jin, K.-K. Wong, J.-C. Chen, and P. Ting, ``Channel estimation for massive MIMO using Gaussian-mixture Bayesian learning,'' {\em IEEE Trans. Wireless Commun.}, vol. 14, no. 3, pp. 1356-1368, Mar. 2015.

\bibitem{J.Li}
J. Li, X. Su, J. Zeng, Y. Zhao, S. Yu, L. Xiao, and X. Xu, ``Codebook design for uniform rectangular arrays of massive antennas,'' Proceedings of IEEE Vehicular Technology Conference, Jun. 2013.

\bibitem{D.Ying}
D. Ying, F. W. Vook, T. A. Thomas, D. J. Love, and A. Ghosh, ``Kronecker product correlation model and limited feedback codebook design in a 3D channel model,'' in Proc. {\em IEEE ICC}, Jun. 2014.

\bibitem{J.Choi2}
J. Choi, Z. Chance, D. J. Love, and U. Madhow, ``Noncoherent trellis-coded quantization: A practical limited feedback technique for massive MIMO systems,'' {\em IEEE Trans. Commun.}, vol. 61, no. 12, pp. 5016-5029, Dec. 2013.

\bibitem{J.Choi3}
J. Choi, D. J. Love, and T. Kim, ``Trellis-extended codebooks and successive phase adjustment: A path from LTE-advanced to FDD massive MIMO systems,'' {\em IEEE Trans. Wireless Commun.}, vol. 14, no. 4, pp. 2007--2016, Apr. 2015.

\bibitem{Emil2}
E. Bj\"{o}rnson, J. Hoydis, M. Kountouris, and M. Debbah, ``Massive MIMO systems with non-ideal hardware: Energy efficiency, estimation, and capacity limits,'' {\em IEEE Trans. Inf. Theory}, vol. 60, no. 11, pp. 7112--7139, Nov. 2014.

\bibitem{X.Zhang0}
X. Zhang, M. Matthaiou, M. Coldrey, and E. Bj\"{o}rnson, ``Energy efficiency optimization in hardware-constrained large-scale MIMO systems,'' in {\it Proc.} {\em IEEE ISWCS}, Aug. 2014, pp. 992--996.

\bibitem{A.Pitarokoilis}
A. Pitarokoilis, S. K. Mohammed, and E. G. Larsson, ``Uplink performance of time-reversal MRC in massive MIMO systems subject to phase noise,'' {\em IEEE Trans. Wireless Commun.}, vol. 14, no. 2, pp. 711--723, Feb. 2015.

\bibitem{K.T.Truong1}
K. T. Truong and R. W. Heath Jr., ``Effects of channel aging in massive MIMO systems,'' {\em J. Commun. Netw.}, vol. 16, no. 4, pp. 338--351, Aug. 2013.

\bibitem{A.K.Papazafeiropoulos1}
A. K. Papazafeiropoulos and T. Ratnarajah, ``Linear precoding for downlink massive MIMO with delayed CSIT and channel prediction,'' in {\it Proc.} {\em IEEE WCNC}, Apr. 2014, pp. 809--914.

\bibitem{A.K.Papazafeiropoulos2}
A. K. Papazafeiropoulos and T. Ratnarajah, ``Uplink performance of massive MIMO subject to delayed CSIT and anticipated channel prediction,'' in {\it Proc.} {\em IEEE ICASSP}, May 2014, pp. 3162--3165.
\bibitem{L.You}
L. You, X. Gao, X. Xia, N. Ma and Y. Peng, ``Massive MIMO transmission with pilot reuse in single cell'', in {\it Proc.} {\em IEEE ICC}, Jun. 2014, pp. 4783--4788.

\bibitem{Emil5}
E. Bj\"{o}rnson, J. Hoydis, M. Kountouris, and M. Debbah, ``Massive MIMO systems with non-ideal hardware: Energy efficiency, estimation, and capacity limits,'' {\em IEEE Trans. Inform. Theory}, vol. 60, no. 11, pp. 7112--7139, Nov. 2014.

\bibitem{Y.Lim}
Y. Lim, C. Chae, and G. Caire, ``Performance analysis of massive MIMO for cell-boundary users,'' accepted to appear in {\em IEEE Trans. Wireless Commun.}, 2015.

\bibitem{A.Muller}
A. Muller, A. Kammoun, E. Bjornson, and M. Debbah, ``Linear precoding based on polynomial expansion: Reducing complexity in massive MIMO,'' available at: arxiv.org/pdf/1310.1806v4.pdf

\bibitem{B.Panzner}
B. Panzner, W. Zirwas, S. Dierks, M. Lauridsen, P. Mogensen, K. Pajukoski, and D. Miao, ``Deployment and implementation strategies for massive MIMO in 5G,'' in {\it Proc} {\em IEEE GLOBECOM}, Dec. 2014, pp. 346--351.

\bibitem{N.Vucic}
N. Vu\v{c}i\'{c} and H. Boche, ``Robust QoS-constrained optimization of downlink multiuser MISO systems,'' {\em IEEE Trans. Signal Process.}, vol. 57, no. 2, pp. 714--725, Feb. 2009.

\bibitem{M.F.Hanif}
M. F. Hanif, L.-N. Tran, A. T\"{o}lli, M. Juntti, and S. Glisic, ``Efficient solutions for weighted sum rate maximization in multicellular networks with channel uncertainties,'' {\em IEEE Trans. Signal Process.}, vol. 61, no. 22, pp. 5659--5674, Nov. \textcolor{blue}{2013}.

\bibitem{E.Bjornson}
E. Bj\"{o}rnson, G. Zheng, M. Bengtsson, and B. Ottersten, ``Robust monotonic optimization framework for multicell MISO systems,'' {\em IEEE Trans. Signal Process.}, vol. 60, no. 5, pp. 2508--2523, May 2012.

\bibitem{A.M.Tulino}
A. M. Tulino and S. Verd\'{u}, ``Random matrix theory and wireless communications,'' {\em Foundations and Trends in Communications and Information Theory}, vol. 1, no. 1, pp. 1--182, Jun. 2004.


%




\bibitem{X.Zhang}
X. Zhang, {\em Matrix Analysis and Applications,} 2nd Ed. Beijing: Tsinghua University Press, 2011.



 \end{thebibliography}

\begin{footnotesize}
 
 \end{footnotesize}

\end{document}